\documentclass[thmsa,12pt]{article}
\normalfont\sffamily
\usepackage{amsfonts}
\usepackage{amssymb}
\usepackage[active]{srcltx}
\usepackage{color}
\usepackage[x11names]{xcolor}
\usepackage{amsmath}
\usepackage{amsthm}
\usepackage{bm} 
\usepackage{graphicx}
\usepackage{textcomp}
\usepackage{a4wide}
\usepackage{multirow}
\usepackage{diagbox}
\usepackage{color}
\usepackage[format=hang]{caption}
\usepackage{subcaption}
\usepackage{url}
\usepackage{xfrac}

\usepackage{tikz}
\usepackage[normalem]{ulem}
\usepackage{comment}
\usepackage{pxfonts}
\usepackage[margin=2cm]{geometry} 
\usepackage{titlesec}

\usepackage{setspace}
\usepackage{footmisc}
\usepackage{accents}
\usepackage{appendix}

\usepackage{array,booktabs,ragged2e}
\newcolumntype{R}[1]{>{\RaggedRight}p{#1}}

\usepackage{cellspace}
\usepackage{hyperref}
\usepackage{hyperxmp}
\hypersetup{
  pdftitle={Equivalent bounded confidence processes},
  pdfauthor={Sascha Kurz},
  pdfsubject={In the bounded confidence model the opinions of a set of agents evolve over discrete time steps. In each round an agent averages the opinion of all agents whose opinions are at most a certain threshold apart. Here we assume that the opinions of the agents are elements of the real line. The details of the dynamics are determined by the initial opinions of the agents, i.e. a starting configuration, and the mentioned threshold -- both admitting uncountably many possibilities. Recently it was conjectured that for each starting configuration the set of thresholds can be partitioned into a finite number of intervals such that the evolution of opinions does not depend on the precise value of the threshold within one of the intervals. So, we may say that, given a starting configuration of initial opinions, there is only a finite number of equivalence classes of bounded confidence processes (and an algorithm to compute them). Here we systematically study different notions of equivalence. Under our broadest notion of equivalence  we can also get rid of the initial starting configuration and end up with a finite number of equivalent bounded confidence processes for each given (finite) number of agents. This allows to precisely study the occurring phenomena for small numbers of agents without the jeopardy of missing interesting cases by performing numerical experiments. We exemplarily study the freezing time, i.e. number of time steps needed until the process stabilizes, and the degree of fragmentation, i.e. the number of different opinions that survive once the process has reached its final state.},
  pdfkeywords={opinion dynamics, bounded confidence, Hegselmann--Krause model, equivalence classes},
  pdfcopyright={Creative Commons Lizenz CC BY 4.0},
  pdflicenseurl={https://creativecommons.org}
}

\newtheorem{lemma}{Lemma}
\newtheorem{conjecture}{Conjecture}

\newtheorem{definition}{Definition}
\newtheorem{remark}{Remark}
\newtheorem{example}{Example}

\def \be {\begin{equation}}
\def \ee {\end{equation}}

\newcommand{\N}{\mathbb{N}}

\newcommand{\Z}{\mathbb{Z}}
\newcommand{\R}{\mathbb{R}}

\newcommand{\cC}{\mathcal{C}}
\newcommand{\cG}{\mathcal{G}}
\newcommand{\cI}{\mathcal{I}}

\newcommand{\cN}{\mathcal{N}}
\newcommand{\cP}{\mathcal{P}}

\begin{document}

\def\title #1{\begin{center}
{\Large {\sc #1}}
\end{center}}
\def\author #1{\begin{center} {#1}
\end{center}}

\setstretch{1.1}

\begin{titlepage}
    \phantomsection \label{Titlepage}
    \addcontentsline{toc}{section}{Title page}

\renewcommand{\thefootnote}{\fnsymbol{footnote}}\addtocounter{footnote}{1}
\title{\sc
Equivalent bounded confidence processes \\ \medskip \  }


\author{Sascha Kurz 
	\\ {\small Dept.\ of Mathematics, University of Bayreuth, 95440 Bayreuth, Germany\\email: sascha.kurz@uni-bayreuth.de}}

\vspace{0.3cm}

\begin{center} {\tt 
\hspace{-1.9em} December 19, 2025 (revised November 23, 2026)
}
\end{center}

\vspace{0.3cm}

\begin{center} {\bf {\sc Abstract}} \end{center}
{\small
In the bounded confidence model the opinions of a set of agents evolve over discrete time steps. In each round an agent averages the opinion of all agents whose opinions are at most a certain threshold apart. Here we assume that the opinions of the agents are elements of the real line. The details of the dynamics are determined by the initial opinions of the agents, i.e.\ a starting configuration, and the mentioned threshold -- both admitting uncountably many possibilities. Recently it was conjectured that for each starting configuration the set of thresholds can be partitioned into a finite number of intervals such that the evolution of opinions does not depend on the precise value of the threshold within one of the intervals. So, we may say that, given a starting configuration of initial opinions, there is only a finite number of equivalence classes of bounded confidence processes (and an algorithm to compute them). Here we systematically study different notions of equivalence. Under our broadest notion of equivalence  we can also get rid of the initial starting configuration and end up with a finite number of equivalent bounded confidence processes for each given (finite) number of agents. This allows to precisely study the occurring phenomena for small numbers of agents without the jeopardy of missing interesting cases by performing numerical experiments. We exemplarily study the freezing time, i.e.\ number of time steps needed until the process stabilizes, and the degree of fragmentation, i.e.\ the number of different opinions that survive once the process has reached its final state. 
}

\vspace{0.2cm}

\begin{description}
{\small
\item[Keywords:] opinion dynamics $\cdot$ bounded confidence $\cdot$ Hegselmann--Krause model $\cdot$ equivalence classes
}
\end{description}


\vspace{2cm}

\vfill
\noindent {\footnotesize I am grateful to Rainer Hegselmann for feedback on earlier drafts and for generously sharing his insights into the bounded confidence model and its pitfalls, as well as being a driving force in slowly but steadily uncovering the mystery of the analytical properties of the bounded confidence model.}

\end{titlepage}

\addtocounter{footnote}{-1}

\setstretch{1.26}

\pagenumbering{arabic}

\section{Introduction}\label{sec_intro}

The Hegselmann--Krause (HK) model \cite{hegselmannkrause2002,krause1997soziale} has become a central framework for studying opinion dynamics. Here, period by period, the agents iteratively adjust their opinions by averaging those that are sufficiently close. More precisely, the interactions are restricted to agents whose opinions differ by no more than a prescribed confidence bound, so that
one only speaks of a bounded confidence model. In the HK model the opinions are updated simultaneously. An alternative
pairwise sequential updating procedure was proposed in \cite{deffuant2000mixing}. In the literature both variants are called bounded confidence (BC) models. While the distinction into the Hegselmann--Krause and the Deffuant--Weisbuch model would be more precise, we will speak of the BC model in the following and restrict our analysis to the synchronous updating procedure. Since its introduction, the BC model has been extensively analyzed from, e.g., dynamical and sociological perspectives, yielding fundamental results on convergence, clustering, and the emergence of consensus or polarization. The defining feature of the simple bounded confidence mechanism gives rise to rich collective behavior and a wide variety of long-term outcomes.

Here we assume that the opinions of the agents are elements of the real line. The details of the dynamics are determined by the initial opinions of the agents, i.e.\ a starting configuration, and the mentioned threshold -- both admitting uncountably many possibilities. Recently it was conjectured that for each starting configuration the set of thresholds can be partitioned into a finite number of intervals such that the evolution of opinions does not depend on the precise value of the threshold within one of the intervals \cite{hegselmann2023bounded}. So, we may say that, given a starting configuration of initial opinions, there is only a finite number of equivalence classes of bounded confidence processes (and an algorithm to compute them). Here we systematically study different notions of equivalence. In our widest notion we can also get rid of the initial starting configuration and end up with a finite number of equivalent bounded confidence processes for each given (finite) number of agents. A so-called influence graph describes the combinatorial interaction between the agents and a sequence of influence graphs combinatorially describes the entire BC-process. This allows to precisely study the occurring phenomena for small numbers of agents without the jeopardy of missing interesting cases by performing numerical experiments. We exemplarily study the freezing time, i.e.\ number of time steps needed until the process stabilizes, and the degree of fragmentation, i.e.\ the number of different opinions that survive once the process has reached its final state. 

The remaining part of this paper is structured as follows. In Section~\ref{sec_preliminaries} we introduce the necessary preliminaries before we dive into the details of influence graphs and linear programming in Section~\ref{sec_influence_graphs}. Given the combinatorial description of BC-processes as sequences of influence graphs and linear programming as a computational tool to decide the existence of a suitable starting configuration we have a machinery to study applications. In Section~\ref{sec_freezing_time} we consider the maximum number of time steps needed to end up in a stable state and in Section~\ref{sec_fragmentation} we study the maximum number of disjoint opinion clusters in the stable state. We close with a brief conclusion and a few open problems in Section~\ref{sec_conclusion}.

\section{Preliminaries}\label{sec_preliminaries}

First we introduce some terminology to precisely state what we mean by a bounded confidence (BC) process:
\begin{enumerate}
  \item[(A)] There is a set $A:=\{1,\dots,n\}$ of $n\in\N_{>0}$ agents.\\[-11mm]
  \item[(B)] Time is discrete; $t\in\N=\{0,1,2,\dots\}$.\\[-11mm]
  \item[(C)] Each agent $i\in A$ starts at time $t=0$ with a certain opinion $x_i(0)\in\R$.\\[-11mm]
  \item[(D)] The \emph{opinion profile} at time $t\in\N$ is given by
             \begin{equation}
                X(t)=\left(x_1(t),x_2(t),\dots,x_n(t)\right)\in\R^n.
             \end{equation}
             We call $X(0)$ the \emph{starting profile} or \emph{starting configuration}.\\[-11mm]
  \item[(E)] Given a \emph{confidence level} $\varepsilon\in\R_{\ge 0}$ the
            \emph{confidence interval} of agent $i\in A$ at time $t\in\N$
            is given by
            \begin{equation}
              V_i(t)=\left[ x_i(t)-\varepsilon,x_i(t)+\varepsilon\right]\subsetneq \R.´
            \end{equation}
  \item[(F)] The set of \emph{$i$-insiders} at time $t\in \N$ for an
             agent $i\in A$ is given by
             \begin{equation}
                I_i(t)=\left\{j\in A\,:\, \left|x_j(t)-x_i(t)\right|\le \varepsilon\right\}=\left\{j\in A\,:\, x_j(t)\in V_i(t)\right\}\subseteq A.
             \end{equation}
             We call its complement
             \begin{equation}
                O_i(t)=A\backslash I_i(t)=\left\{j\in A\,:\, \left|x_j(t)-x_i(t)\right|> \varepsilon\right\}=\left\{j\in A\,:\, x_j(t)\notin V_i(t)\right\}\subsetneq A
             \end{equation}
             the set of \emph{$i$-outsiders} at time $t\in\N$.\\[-11mm]
  \item[(G)] The \emph{updating rule} for the opinions is given by
             \begin{equation}
               \label{eq_update_rule}
               x_i(t+1)= \frac{1}{\# I_i(t)}\cdot\sum_{j\in I_i(t)} x_j(t)
             \end{equation}
             for all $i\in A$ and all $t\in \N$.\\[-11mm]
  \item[(H)] The \emph{influence graph} at time $t\in\N$ is denoted as $\cG(t)=(A,E(t))$, where the vertex set is given by the set $A$ of agents and the edge set is given by
  \begin{equation}
    E(t):=\big\{\{i,j\}\,:\, i,j\in A,i\neq j, \left|x_i(t)-x_j(t)\right|\le\varepsilon \big\}.
  \end{equation}
\end{enumerate}

We remark that some authors assume that the starting opinions $x_i(0)$, and in general $x_i(t)$, are contained in the unit interval $[0,1]$. We will see in a moment that both definitions lead to equivalent BC-processes, while our choice allows greater flexibility in certain situations. The confidence level $\varepsilon$ is sometimes also called confidence radius. If we restrict the $x_i(t)$ to $[0,1]$, then we can also assume $\varepsilon\in[0,1]$ w.l.o.g. Of course, the choice $\varepsilon=0$ leads to a rather uninteresting BC-process where we have $x_i(t)=x_i(0)$ for all $t\in\N$ and all $i\in A$. We will nevertheless include the case $\varepsilon=0$ in our definition to make the presentation smoother at some places. Note that we always have $i\in I_i(t)$, i.e.\ $I_i(t)$ is non-empty, so that the updating rule (\ref{eq_update_rule}) is well defined. While the influence graph is not used in a majority of the papers on the BC model, it can be traced back to \cite{kurz2014long,kurz2015optimal}. In \cite{wedin2015hegselmann} the authors speak of a receptivity graph. We will see that the notion of the influence graphs is an important combinatorial invariant of a BC-process. We remark that we have
\begin{equation}
  \label{eq_insider_graph}
  I_i(t)=\{i\}\cup \left\{j\in A\backslash\{i\}\,:\, \{i,j\}\in E(t)\right\},
\end{equation}
so that the updating rule (\ref{eq_update_rule}) can also be expressed in terms of $\cG(t)$.

\begin{remark}
  In this paper we restrict ourselves to the one dimensional bounded confidence model where the agents' opinions $x_i(t)$ are real numbers. Of course this model has been generalized to higher dimensions where $x_i(t)\in\R^l$ or a subset thereof. In order to generalize the notions of the confidence interval, $i$-insiders, $i$-outsiders, and the influence graph,
  we need to choose a \emph{metric} $d$ on $\R^l$ (or $\R$). With this the \emph{confidence region} of agent $i\in A$ at time $t\in \N$ is given by $V_i(t):=\left\{y\in\R^l\,:\, d\!\left(y,x_i(t)\right)\le\varepsilon\right\}\subseteq\R^l$. If the sets of $i$-insiders and $i$-outsiders at time $t$ are defined via the confidence interval (region) $V_i(t)$ then we can keep the definition. Otherwise we set $I_i(t)=\left\{j\in A\,: d\!\left(x_j(t),x_i(t)\right)\le \varepsilon\right\}\subseteq A$ and $O_i(t)=\left\{j\in A\,: d\!\left(x_j(t),x_i(t)\right)> \varepsilon\right\}\subseteq A$. Finally, the edge set of the influence graph at time $t$ is given by $E(t)=\left\{\{i,j\}\,:\, i,j\in A,i\neq j,d\!\left(x_i(t),x_j(t)\right)\le \varepsilon\right\}$, so that Equation~(\ref{eq_insider_graph}) remains valid.
  For the choice of the metric $d$ one usually takes one induced by a \emph{norm} $\Vert\cdot\Vert$, i.e.\ $d(x,y):=\Vert x-y\Vert$.
  On the real line $\R$ this leaves the choice $\Vert x\Vert:=\alpha\cdot |x|$ for some $\alpha>0$ only.
  By rescaling the confidence level $\varepsilon$ we may assume $\alpha=1$ w.l.o.g.\ For $l\ge 2$ there are infinitely many norms with corresponding non-similar unit balls.
  Famous examples are the \emph{Euclidean norm} $\left\Vert \left(x_1,\dots,x_l\right)\right\Vert_2:=\sqrt{\sum_{i=1}^l x_i^2}$,
  the \emph{$1$-norm}  $\left\Vert \left(x_1,\dots,x_l\right)\right\Vert_1:=\sum_{i=1}^l \left|x_i\right|$
  (inducing the \emph{Taxicab geometry} a.k.a.\ \emph{Manhattan geometry}),
  and the \emph{maximum norm} (or \emph{Chebyshev norm}) $\left\Vert \left(x_1,\dots,x_l\right)\right\Vert_\infty:=\max \left\{ \left|x_i\right|\,:\, 1\le i\le l\right\}$. Some of our later statements are also valid in this more general setting. Here we will give remarks and refer to the higher dimensional BC model or higher dimensional BC-processes, where the latter notion is introduced below.
\end{remark}

Based on (A)--(H) we can now define that the sequence of opinion profiles $\big(X(t)\big)_{t\in\N}$ is a BC-process if it is generated by the updating rule (\ref{eq_update_rule}). Noting that the number of entries of $X(0)$ encodes the number $n=\#A$ of agents, we observe that a BC-process, which is an $n$-dimensional dynamical system, is uniquely described by the starting configuration $X(0)$ and the confidence level $\varepsilon$, so that we also speak of a BC-process $\cP=\left(X(0),\varepsilon\right)$. So formally, we may write $V_i(t;X(0),\varepsilon)$, $I_i(t;X(0),\varepsilon)$, $O_i(t;X(0),\varepsilon)$, $\cG(t;X(0),\varepsilon)$ or $V_i(t;\cP)$, $I_i(t;\cP)$, $O_i(t;\cP)$, $\cG(t;\cP)$. However, we prefer the more handy abbreviations whenever the underlying BC-process $\cP$ is clear from the context.

An important property of BC-processes is that they stabilize in a finite number of time steps, see e.g.\ \cite[Theorem 9]{dittmer2001consensus} and, in a more general setting covering higher dimensional opinions, \cite{lorenz2005stabilization}. Formally, the freezing time $F(\cP)$ is the smallest integer $t$ such that $X(t+1)=X(t)$, which obviously implies $X(t')=X(t)$ for all $t'\ge t$. Interestingly enough, the freezing time of a BC-process can be upper bounded by a polynomial in its number $n$ of agents. To this end, we define $F(n)$ as the maximum freezing time of a BC-process with $n$ agents. In \cite{wedin2015quadratic} a sequence of examples with a quadratic freezing time was constructed. So, we have $F(n)\in\Omega\!\left(n^2\right)$.\footnote{We write $f(n)\in \Omega(g(n))$ if there exist constants $n_0\in\N$ and $c\in\R_{>0}$ such that $f(n)\ge c\cdot g(n)$ for all $n\ge n_0$.} A cubic upper bound was shown in \cite{bhattacharyya2013convergence,mohajer2013convergence}. So, we e.g.\ have $F(n)<4n\!\left(n^2+1\right)\in O\!\left(n^3\right)$.\footnote{We write $f(n)\in O(g(n))$ if there exist constants $n_0\in\N$ and $c\in\R_{>0}$ such that $f(n)\le c\cdot g(n)$ for all $n\ge n_0$. Earlier, the upper bounds $F(n)\in O\!\left(n^5\right)$ and $F(n)\in O\!\left(n^4\right)$ were shown in \cite{martinez2007synchronous} and \cite{touri2011discrete}, respectively.} Note that the upper bound is independent of the starting configuration $X(0)$ and the confidence level $\varepsilon$. For higher dimensional BC-processes in $\R^l$ the freezing time is upper bounded by $O\!\left(n^{10} l^2\right)$ \cite[Theorem 4.4]{bhattacharyya2013convergence}.

\medskip

\begin{example}(Cf.~\cite[Theorem 3(i)]{krause2000discrete})\label{example_1}
  Consider the BC-process $\cP=(X(0),\varepsilon)$ given by $X(0)=(2,3,4,5)$ and $\varepsilon=1$. Evaluating the updating rule (\ref{eq_update_rule}) we compute $X(1)=\left(\tfrac{5}{2},3,4,\tfrac{9}{2}\right)$, $X(2)=\left(\tfrac{11}{4},\tfrac{19}{6},\tfrac{23}{6},\tfrac{17}{4}\right)$, $X(3)=\left(\tfrac{71}{24},\tfrac{13}{4},\tfrac{15}{4},\tfrac{97}{24}\right)$, $X(4)=\left(\tfrac{239}{72},\tfrac{7}{2},\tfrac{7}{2},\tfrac{265}{72}\right)$, and $X(5)=\left(\tfrac{7}{2},\tfrac{7}{2},\tfrac{7}{2},\tfrac{7}{2}\right)$, so that $F(\cP)=5$. The corresponding sequence of influence graphs $\cG(t)$ is given by:\\
  $\cG(0)$: 
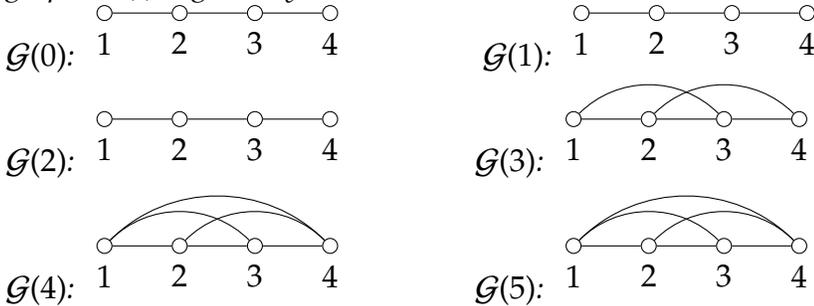
\begin{tikzpicture}[node/.style={circle, draw, inner sep=2pt}]
    \node[node,label=below:$1$] (v1) at (0,0) {};
    \node[node,label=below:$2$] (v2) at (1,0) {};
    \node[node,label=below:$3$] (v3) at (2,0) {};
    \node[node,label=below:$4$] (v4) at (3,0) {};
    \draw (v1) -- (v2) -- (v3) -- (v4);
  \end{tikzpicture} \quad\quad\quad\quad
  $\cG(1)$:  \begin{tikzpicture}[node/.style={circle, draw, inner sep=2pt}]
    \node[node,label=below:$1$] (v1) at (0,0) {};
    \node[node,label=below:$2$] (v2) at (1,0) {};
    \node[node,label=below:$3$] (v3) at (2,0) {};
    \node[node,label=below:$4$] (v4) at (3,0) {};
    \draw (v1) -- (v2) -- (v3) -- (v4);
  \end{tikzpicture}\\
  $\cG(2)$:  \begin{tikzpicture}[node/.style={circle, draw, inner sep=2pt}]
    \node[node,label=below:$1$] (v1) at (0,0) {};
    \node[node,label=below:$2$] (v2) at (1,0) {};
    \node[node,label=below:$3$] (v3) at (2,0) {};
    \node[node,label=below:$4$] (v4) at (3,0) {};
    \draw (v1) -- (v2) -- (v3) -- (v4);
  \end{tikzpicture}\quad\quad\quad\quad
  $\cG(3)$:  \begin{tikzpicture}[node/.style={circle, draw, inner sep=2pt}]
    \node[node,label=below:$1$] (v1) at (0,0) {};
    \node[node,label=below:$2$] (v2) at (1,0) {};
    \node[node,label=below:$3$] (v3) at (2,0) {};
    \node[node,label=below:$4$] (v4) at (3,0) {};
    \draw (v1) -- (v2) -- (v3) -- (v4);
    \draw[bend left=45] (v1) to (v3);
    \draw[bend left=45] (v2) to (v4);
  \end{tikzpicture}\\
  $\cG(4)$:  \begin{tikzpicture}[node/.style={circle, draw, inner sep=2pt}]
    \node[node,label=below:$1$] (v1) at (0,0) {};
    \node[node,label=below:$2$] (v2) at (1,0) {};
    \node[node,label=below:$3$] (v3) at (2,0) {};
    \node[node,label=below:$4$] (v4) at (3,0) {};
    \draw (v1) -- (v2) -- (v3) -- (v4);
    \draw[bend left=45] (v1) to (v3);
    \draw[bend left=45] (v2) to (v4);
    \draw[bend left=45] (v1) to (v4);
  \end{tikzpicture}\quad\quad\quad\quad
  $\cG(5)$:  \begin{tikzpicture}[node/.style={circle, draw, inner sep=2pt}]
    \node[node,label=below:$1$] (v1) at (0,0) {};
    \node[node,label=below:$2$] (v2) at (1,0) {};
    \node[node,label=below:$3$] (v3) at (2,0) {};
    \node[node,label=below:$4$] (v4) at (3,0) {};
    \draw (v1) -- (v2) -- (v3) -- (v4);
    \draw[bend left=45] (v1) to (v3);
    \draw[bend left=45] (v2) to (v4);
    \draw[bend left=45] (v1) to (v4);
  \end{tikzpicture}\\
  We have $\cG(t)=\cG(5)$ and $X(t+1)=X(t)$ for all $t\ge 5$.
\end{example}

\bigskip

Next we collect a few basic analytic properties:\\[-8mm]
\begin{lemma}$\,$\\[-9mm]\label{lemma_basic}
\begin{itemize}
 \item[(1)] If $x_i(t)\le x_j(t)$ for some $i,j\in A$ and $t\in \N$, then we have $x_i(t')\le x_j(t')$ for all $t'\ge t$, cf.\ \cite[Lemma 2]{krause2000discrete}.\\[-10mm]
 \item[(2)] For all $t\in\N$ we have
 $\min\{x_i(t)\,:\, i\in A\} \le \min\{x_i(t+1)\,:\, i\in A\}$
 and $\max\{x_i(t)\,:\, i\in A\} \ge \max\{x_i(t+1)\,:\, i\in A\}$.\\[10mm]
 \item[(3)] For all $i,j\in A$ and $t\in \N$ we have that $j\in I_i(t)$ is equivalent to $i\in I_j(t)$.\\[-10mm]
 \item[(4)] If $\#\{i\in A\,:\, x_i(t)=y\}=\#\{i\in A\,:\, x_i(t)=2m-y\}$ for all $y\le m$, then we say that $X(t)$ is \emph{mirror symmetric} at $m\in\R$. So, if $X(t)$ is mirror symmetric at $m$, so is $X(t')$ for all $t'\ge t$.\\[-10mm]
 \item[(5)] Let $\cP=(X(0),\varepsilon)$ be an arbitrary BC-process. If $X'(0)=\alpha X(0)+\beta$ for some $\alpha,\beta\in\R$ with $\alpha\neq 0$, i.e.\ $x_i'(0)=\alpha x_i(0)+\beta$ for all $i\in A$, and $\varepsilon'=|\alpha|\varepsilon$, then we have $X'(t)=\alpha X(t)+\beta$ for all $t$ for the BC-process $\cP'=(X'(0),\varepsilon')$. We also write $\cP'=\alpha\cP+\beta$.\\[-10mm]
  \item[(6)] Let $\cP=(X(0),\varepsilon)$ and $\cP'=(X'(0),\varepsilon')$ be two BC-processes with equal starting configuration, i.e.\ $X(0)=X'(0)$. If $\cG(i)=\cG'(i)$ for all $0\le i< t$ and $\cG(t)\neq\cG'(t)$, then we have $X(i)=X'(i)$ for all $0\le i\le t$ and $X(t+1)\neq X'(t+1)$.
\end{itemize}
\end{lemma}
We remark that statement (3) of Lemma~\ref{lemma_basic} is also valid for higher dimensional BC-processes in $\R^l$, as well as statement~(5) if we choose $\beta\in\R^l$. The part $X(t+1)\neq X'(t+1)$ in statement~(6) is only valid in $\R$.\footnote{Counterexamples are available from the author upon request.}
\begin{proof}[Proof of Lemma~\ref{lemma_basic}.(6)]
  If $X(s)=X'(s)$ and $G(s)=G'(s)$, then Equation~(\ref{eq_insider_graph}) implies $I_k(s)=I'_k(s)$ for all $k$, so that $X(s+1)=X'(s+1)$ is given by Equation~(\ref{eq_update_rule}). Induction gives $X(s)=X'(s)$ for all $s\le t$.

  At time $t$ both graphs arise from the same distances $\left|x_i(t)-x_j(t)\right|$ by thresholding at $\varepsilon$
  resp.\ $\varepsilon'$. Since $G(t)\neq G'(t)$, we may assume $\varepsilon'<\varepsilon$, so
  $I'_k(t)\subseteq I_k(t)$ for all $k$. Choose $i$ with
  $I_i(t)\setminus I'_i(t)\neq\emptyset$ and $x_i(t)$ minimal. For $k\in I_i(t)\setminus I'_i(t)$ we have
  $\varepsilon'<\left|x_k(t)-x_i(t)\right|\le\varepsilon$, hence $i\in I_k(t)\setminus I'_k(t)$, and minimality yields
  $x_k(t)>x_i(t)+\varepsilon'$. Thus every opinion in $I_i(t)\setminus I'_i(t)$ exceeds every
  opinion in $I'_i(t)\subseteq\left[x_i(t)-\varepsilon',x_i(t)+\varepsilon'\right]$. Therefore 
  $x_i(t+1)>x'_i(t+1)$.
\end{proof}
Due to the first property we will always assume the ordering $x_1(0)\le x_2(0)\le\dots\le x_n(0)$ in the following. Defining the \emph{width} of a given BC-process $\cP=(X(0),\varepsilon)$ as
\begin{equation}
  \omega(t)=\max\{x_i(t)\,:\, i\in A\}-\min\{x_i(t)\,:\, i\in A\},
\end{equation}
the second property states that the width $\omega(t)$ is weakly decreasing in $t$. The \emph{initial width} and the \emph{final width} of $\cP$ are given by $\omega(0)$ and $\omega(F(\cP))$, respectively. In Example~\ref{example_1} the initial width equals three and the final width is equal to zero. Whenever the final width equals zero we speak of \emph{consensus}, i.e.\ all agents have the same opinion in the final state, and of \emph{fragmentation} otherwise. Applying an affine transform with constants $\alpha$ and $\beta$ results in a final (and initial) width multiplied by $|\alpha|$.
The third property is the reason why we can consider an undirected influence graph $\cG(t)$. In our BC-process in Example~\ref{example_1} we have that $X(0)$ is mirror symmetric at $m=\tfrac{7}{2}$. Since we already know that it reaches a consensus, the fourth property implies that the final unique opinion equals $m$. A mirror symmetry of $X(t)$ also implies that the corresponding influence graph $\cG(t)$ has an automorphism of order two. However, $\cG(t)$ can have many more symmetries than $X(t)$.
Due to the fifth property we say that two BC-processes $\cP$ and $\cP'$ are \emph{affine equivalent} if there exist constants $\beta\in\R$ and $\alpha\in\R\backslash\{0\}$ such that $\cP'=\alpha\cP+\beta$. The BC-process in Example~\ref{example_1} is affine equivalent to the BC-process $\cP'=(X'(0),\varepsilon')$ with $X'(0)=\left(0,\tfrac{1}{3},\tfrac{2}{3},1\right)$ and $\varepsilon'=\tfrac{1}{3}$. So, when only considering BC-processes up to affine equivalence we can always restrict the opinions $x_i(t)$ to the unit interval $[0,1]$. Alternatively, for $\varepsilon>0$ we can always assume $\varepsilon=1$. We remark that we will use the sixth property in the proof of Lemma~\ref{lemma_epsilon_equivalence_structure}.

Certainly, we would not say that two affine equivalent BC-processes $\cP$ and $\cP'$ are essentially different, while we mostly have $\cP\neq \cP'$. We write $\cP\cong\cP'$ for two equivalent BC-processes. Since we will introduce more notions of equivalence we will write  $\cP\cong_A\cP'$ when we refer to affine equivalence and the equivalence relation is not clear from the context. In general the \emph{equivalence class} of $\cP$ consists of the set of all BC-processes that are equivalent to $\cP$. For affine equivalence the equivalence classes contain infinitely many elements, parameterized by the constants $\alpha$ and $\beta$, while there are still infinitely many equivalence classes if $n>1$.

In order to introduce another equivalence relation we observe that in Example~\ref{example_1} we could modify the confidence level $\varepsilon$ to any value in $\big[1,\tfrac{13}{12}\big)$ without changing any of the opinion profiles $X(t)$. For $\varepsilon\in [0,1)$ we obtain a BC-process with freezing time $0$ and in Example~\ref{example_2} we will see that the corresponding BC-process is essentially different for $\varepsilon=\tfrac{13}{12}$.
\begin{definition}
  Given two BC-processes $\cP=(X(0),\varepsilon)$ and $\cP'=(X'(0),\varepsilon')$ we say that $\cP$ and $\cP'$ are \emph{$\varepsilon$-equivalent}, denoted by $\cP\cong_\varepsilon\cP'$, iff $X(t)=X'(t)$ for all $t\in \N$.
\end{definition}
\begin{example}\label{example_2}
  Consider the BC-process $\cP=(X(0),\varepsilon)$ given by $X(0)=(2,3,4,5)$ and $\varepsilon=\tfrac{13}{12}$. Evaluating the updating rule (\ref{eq_update_rule}) we compute $X(1)=\left(\tfrac{5}{2},3,4,\tfrac{9}{2}\right)$, $X(2)=\left(\tfrac{11}{4},\tfrac{19}{6},\tfrac{23}{6},\tfrac{17}{4}\right)$, $X(3)=\left(\tfrac{13}{4},\tfrac{7}{2},\tfrac{7}{2},\tfrac{15}{4}\right)$, and $X(4)=\left(\tfrac{7}{2},\tfrac{7}{2},\tfrac{7}{2},\tfrac{7}{2}\right)$, so that $F(\cP)=4$. The corresponding sequence of influence graphs $\cG(t)$ is given by:\\
  $\cG(0)$: 
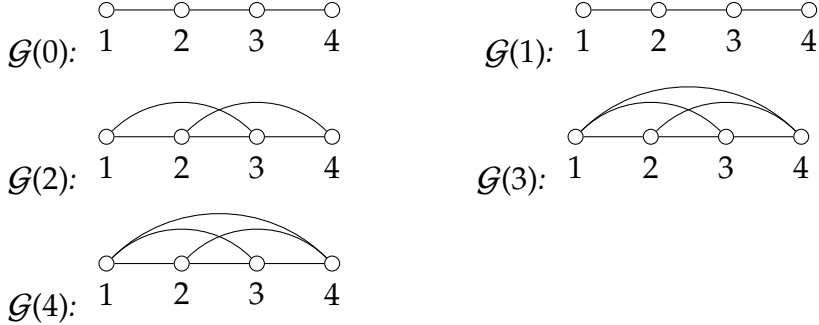
\begin{tikzpicture}[node/.style={circle, draw, inner sep=2pt}]
    \node[node,label=below:$1$] (v1) at (0,0) {};
    \node[node,label=below:$2$] (v2) at (1,0) {};
    \node[node,label=below:$3$] (v3) at (2,0) {};
    \node[node,label=below:$4$] (v4) at (3,0) {};
    \draw (v1) -- (v2) -- (v3) -- (v4);
  \end{tikzpicture} \quad\quad\quad\quad
  $\cG(1)$:  \begin{tikzpicture}[node/.style={circle, draw, inner sep=2pt}]
    \node[node,label=below:$1$] (v1) at (0,0) {};
    \node[node,label=below:$2$] (v2) at (1,0) {};
    \node[node,label=below:$3$] (v3) at (2,0) {};
    \node[node,label=below:$4$] (v4) at (3,0) {};
    \draw (v1) -- (v2) -- (v3) -- (v4);
  \end{tikzpicture}\\
  $\cG(2)$:  \begin{tikzpicture}[node/.style={circle, draw, inner sep=2pt}]
    \node[node,label=below:$1$] (v1) at (0,0) {};
    \node[node,label=below:$2$] (v2) at (1,0) {};
    \node[node,label=below:$3$] (v3) at (2,0) {};
    \node[node,label=below:$4$] (v4) at (3,0) {};
    \draw (v1) -- (v2) -- (v3) -- (v4);
    \draw[bend left=45] (v1) to (v3);
    \draw[bend left=45] (v2) to (v4);
  \end{tikzpicture}\quad\quad\quad\quad
  $\cG(3)$:  \begin{tikzpicture}[node/.style={circle, draw, inner sep=2pt}]
    \node[node,label=below:$1$] (v1) at (0,0) {};
    \node[node,label=below:$2$] (v2) at (1,0) {};
    \node[node,label=below:$3$] (v3) at (2,0) {};
    \node[node,label=below:$4$] (v4) at (3,0) {};
    \draw (v1) -- (v2) -- (v3) -- (v4);
    \draw[bend left=45] (v1) to (v3);
    \draw[bend left=45] (v2) to (v4);
    \draw[bend left=45] (v1) to (v4);
  \end{tikzpicture}\\
  $\cG(4)$:  \begin{tikzpicture}[node/.style={circle, draw, inner sep=2pt}]
    \node[node,label=below:$1$] (v1) at (0,0) {};
    \node[node,label=below:$2$] (v2) at (1,0) {};
    \node[node,label=below:$3$] (v3) at (2,0) {};
    \node[node,label=below:$4$] (v4) at (3,0) {};
    \draw (v1) -- (v2) -- (v3) -- (v4);
    \draw[bend left=45] (v1) to (v3);
    \draw[bend left=45] (v2) to (v4);
    \draw[bend left=45] (v1) to (v4);
  \end{tikzpicture}\\
  We have $\cG(t)=\cG(4)$ and $X(t+1)=X(t)$ for all $t\ge 4$.
\end{example}
\begin{lemma}\label{lemma_finite_number_euqivalence_classes_1}
  For each starting configuration $X(0)\in\R^n$ there exists a finite number of equivalence classes with respect to $\cong_\varepsilon$, i.e.\ there exist $\varepsilon_1,\dots,\varepsilon_l\in\R_{\ge 0}$ such that for each $\varepsilon'\in\R_{\ge 0}$ there exists an index $1\le h\le l$ with $\big(X(0),\varepsilon'\big)\cong_\varepsilon\big(X(0),\varepsilon_h\big)$.
\end{lemma}
\begin{proof}
  As mentioned before, the updating rule (\ref{eq_update_rule}) can be expressed solely in terms of the influence graph $\cG(t)$, see Equation~(\ref{eq_insider_graph}). Since the number of steps till freezing is upper bounded and for each step there is a finite number of possible influence graphs, the result follows.
\end{proof}

Since the number of steps for $n$ agents is upper bounded by $4n(n^2+1)$ and the number of graphs is upper bounded by $2^{n\choose 2}$, a first crude upper bound for the number of equivalence classes for $n$ agents with respect to the equivalence relation $\cong_\varepsilon$ is given by $\left(2^{n\choose 2}\right)^{4n(n^2+1)}$, which also holds for higher dimensional BC-processes in $\R^l$ when $4n(n^2+1)$ is replaced by some upper bound on the freezing time in $\R^l$.

The notion of $\varepsilon$-equivalence appears implicitly in \cite{hegselmann2023bounded}. There the author considers so-called $\varepsilon$-switches. Loosely speaking, $\varepsilon$-switches are specific values where the BC-process changes when the confidence level is increased to them. They allow a precise characterization of the equivalence classes:
\begin{lemma}
  \label{lemma_epsilon_equivalence_structure}
  For each starting configuration $X(0)\in\R^n$ there exists a finite number of real values $0<\varepsilon_1<\dots<\varepsilon_l<\infty$ such that for each $\varepsilon'\in \big[\varepsilon_j,\varepsilon_{j+1}\big)$, where $0\le j\le l$, $\varepsilon_0=0$, and $\varepsilon_{l+1}=\infty$, we have $\big(X(0),\varepsilon'\big)\cong_\varepsilon \big(X(0),\varepsilon_j\big)$ and $\big(X(0),\varepsilon'\big)\not\cong_\varepsilon\big(X(0),\varepsilon''\big)$ for all $\varepsilon''\in\R_{\ge 0}\backslash \big[\varepsilon_j,\varepsilon_{j+1}\big)$.
\end{lemma}
\begin{proof}
  We will prove the statement recursively. By Lemma~\ref{lemma_basic}.(6)
  equal graphs give $X'(i)=X''(i)$ and different graphs give $X'(i)\ne X''(i)$. In the beginning let $l(0)=0$, $\varepsilon_0^{(0)}=0$, $\varepsilon_1^{(0)}=\infty$, and $\cI^{(0,0)}=[0,\infty)$. For each step $i>0$ we start with $\left(\varepsilon_0^{(i)},\dots,\varepsilon_{l(i-1)+1}^{(i)}\right)=\left(\varepsilon_0^{(i-1)},\dots,\varepsilon_{l(i-1)+1}^{(i-1)}\right)$ and insert further values as described below by looping over all intervals $\cI^{(i,j)}=\big[\varepsilon_j^{(i-1)},\varepsilon_{j+1}^{(i-1)}\big)$ for $0\le j\le l(i-1)$. Given an interval $\cI^{(i,j)}$ let
  $$
    \cN^{(i,j)}=\left(\left\{\left| x_a(i-1)-x_b(i-1)\right|\,:\, a,b\in A\right\}\cap \cI^{(i,j)}\right)\backslash\left\{\varepsilon_j^{(i-1)}\right\},
  $$
  i.e.\ the set of different distances between the opinions of a pair of agents at time $i-1$ within the interval $\cI^{(i,j)}$ excluding its left boundary point. Between $\varepsilon_j^{(i-1)}$ and $\varepsilon_{j+1}^{(i-1)}$ we insert the elements of $\cN^{(i,j)}$ in an increasing order.

  We write $(X'(0),\varepsilon')\cong_\varepsilon^t (X''(0),\varepsilon'')$ iff $X'(t')=X''(t')$ for all $0\le t'\le t$.
  By induction over $i$ we prove the existence of real values $0<\varepsilon_1^{(i)}<\dots<\varepsilon_{l(i)}^{(i)}<\infty$
  such that for each $\varepsilon'\in \big[\varepsilon_j^{(i)},\varepsilon_{j+1}^{(i)}\big)$,
  where $0\le j\le l(i)$, $\varepsilon_0^{(i)}=0$, and $\varepsilon_{l(i)+1}^{(i)}=\infty$,
  we have $\big(X(0),\varepsilon'\big)\cong_\varepsilon^i \big(X(0),\varepsilon_j^{(i)}\big)$ and $\big(X(0),\varepsilon'\big)\not\cong_\varepsilon^i\big(X(0),\varepsilon''\big)$ for all $\varepsilon''\in\R_{\ge 0}\backslash \big[\varepsilon_j^{(i)},\varepsilon_{j+1}^{(i)}\big)$.

  For the induction start $i=0$ the statement is obviously true. For the induction step we observe that we have constructed a sequence $0<\varepsilon_1^{(i)}<\dots<\varepsilon_{l(i)}^{(i)}<\infty$. If $\varepsilon'\in \cI^{(i,j)}$ and $\varepsilon''\in \cI^{(i,j')}$ with $j\neq j'$ then we have $(X(0),\varepsilon')\not\cong_\varepsilon^i (X(0),\varepsilon'')$ due to the induction hypothesis $(X(0),\varepsilon')\not\cong_\varepsilon^{i-1} (X(0),\varepsilon'')$.
  If $j=j'$, then the induction hypothesis gives $(X(0),\varepsilon')\cong_\varepsilon^{i-1} (X(0),\varepsilon'')$.
  By construction we have $\cG'(i-1)=\cG''(i-1)$, denoting the influence graphs for $(X(0),\varepsilon')$ and $(X(0),\varepsilon'')$,
  iff there exists an index $0\le h\le l(i)$ such that
  $\varepsilon',\varepsilon''\in\big[\varepsilon_h^{(i)},\varepsilon_{h+1}^{(i)}\big)$, which proves the induction step.

  It remains to observe that after at most $F(n)$ steps the BC-processes stabilize and we have $\cN^{(i,j)}=\emptyset$ for all $i\ge F(n)+2$. So, setting $t=F(n)+1$ and $\left(\varepsilon_0,\varepsilon_1,\dots,\varepsilon_{l+1}\right)=\left(\varepsilon_0^{(t)},\varepsilon_1^{(t)},\dots,\varepsilon_{l(t)+1}^{(t)}\right)$ we obtain the stated result since we have $\cP\cong_\varepsilon\cP'$ iff $\cP\cong_\varepsilon^t\cP'$.
\end{proof}
In words Lemma~\ref{lemma_epsilon_equivalence_structure} says that the $\varepsilon$-equivalence classes are given by $\big[0,\varepsilon_1\big)$, $\big[\varepsilon_1,\varepsilon_2\big)$, $\dots$, $\big[\varepsilon_{l-1},\varepsilon_l\big)$, $\big[\varepsilon_l,\infty\big)$, i.e.\ confidence levels within the same interval correspond to equivalent BC-processes and confidence levels from different intervals correspond to non-equivalent BC-processes. Those values $\varepsilon_i$ with $1\le i\le l$ are called $\varepsilon$-switches in \cite{hegselmann2023bounded}. The stated proof describes an algorithm for explicitly computing these $\varepsilon$-switches. We remark that this algorithm is the breadth first  version of the depth first type algorithm described in \cite{hegselmann2023bounded}. Lemma~\ref{lemma_finite_number_euqivalence_classes_1} and Lemma~\ref{lemma_epsilon_equivalence_structure} prove the conjecture in \cite[Analytical Note 6]{hegselmann2023bounded}. Lemma~\ref{lemma_epsilon_equivalence_structure} completely characterizes the structure of the equivalence classes w.r.t.\ $\cong_\varepsilon$. We remark that the proof of Lemma~\ref{lemma_epsilon_equivalence_structure} also implies that the $\varepsilon$-switches are rational numbers if $X(0)\in\mathbb{Q}^n$. Otherwise the proof still works but on the algorithmic side we have to deal with the problem of how to represent and compute with non-rational real numbers. While we do not state any lower bound on the minimal length of the intervals $\big[\varepsilon_j,\varepsilon_{j+1}\big)$ (in terms of the input data $X(0)$), we remark that all of these intervals have a non-zero length at the very least. So, in principle one may find all different BC-processes by slowly increasing $\varepsilon$ iteratively by a sufficiently small constant $\Delta>0$.\footnote{This is e.g.\ different in the context of voting where certain scoring rules only exist for singular non-rational choices of parameters if the number of alternatives is larger than two, see e.g.\ \cite{kurz2023art} for a description of the application context.} For the BC-process in Example~\ref{example_1} the $\varepsilon$-switches are given by
$1$, $\tfrac{13}{12}$, $\tfrac{3}{2}$, $2$, $3$.
Clearly, an affine transform $\alpha X(0)+\beta$ transforms the $\varepsilon$-switches by multiplication with $|\alpha|$. Actually, our last interval in Lemma~\ref{lemma_epsilon_equivalence_structure} has infinite length. However, the initial width is the largest value for $\varepsilon$ that we actually need. So, if $X(0)\in [0,1]^n$, then we do not need to consider confidence levels larger than $1$ and we may replace $\varepsilon_{l+1}$ by any number larger than $1$.

The proof of Lemma~\ref{lemma_finite_number_euqivalence_classes_1} suggests an even coarser notion of equivalence.
\begin{definition}
  Two BC-processes $\cP'=(X'(0),\varepsilon')$ and $\cP''=(X''(0),\varepsilon'')$ are called \emph{graph equivalent}, in symbols $\cP'\cong_{\cG}\cP''$, iff their influence graphs coincide, i.e.\ $\cG'(t)=\cG''(t)$ for all $t\in\N$.
\end{definition}
\begin{example}
  \label{example_3}
  Let us fix the confidence level at $\varepsilon=1$ and describe the set of starting profiles $X(0)\in\R^4$, assuming the affine normalization $x_1(0)=0$,  which is graph equivalent to the BC-process from Example~\ref{example_1}. As a parameterization we choose $X(0)=(0,x,x+y,x+y+z)$, where $x,y,z\in\R_{\ge 0}$. The specified influence graph $\cG(0)$ gives the constraints $x\le 1$, $y\le 1$, $z\le 1$, $x+y>1$, and $y+z>1$. Note that the constraint $x+y+z>1$ for the missing edge between agent~$1$ and agent~$4$ is implied by the other constraints. Using $\cG(0)$ we compute $X(1)=\left(\tfrac{x}{2},\tfrac{2x+y}{3},x+\tfrac{2y+z}{3},x+y+\tfrac{z}{2}\right)$. With this, the specified influence graph $\cG(1)$ yields the additional constraints $3x+4y+2z>6$ and $2x+4y+3z>6$. Using $\cG(1)$ we compute $X(2)=\left(\tfrac{7x+2y}{12},\tfrac{13x+6y+2z}{18},\tfrac{16x+12y+5z}{18},x+\tfrac{10y+5z}{12}\right)$, so that $\cG(2)$ implies the constraints $11x+18y+10z>36$ and $10x+18y+11z>36$. Using $\cG(2)$ we compute $X(3)=\left(\tfrac{47x+18y+4z}{72},\tfrac{79x+42y+14z}{108},\tfrac{94x+66y+29z}{108},\tfrac{68x+54y+25z}{72}\right)$, so that $\cG(3)$ implies the constraints $47x+78y+46z \le 216$ and $46x+78y+47z\le 216$, which are automatically satisfied, as well as the additional constraint $7x+12y+7z>24$. Using $\cG(3)$ we compute
  $X(4)=\left(\tfrac{487x+270y+98z}{648},\tfrac{691x+432y+173z}{864},\tfrac{691x+432y+173z}{864},\tfrac{550x+378y+161z}{648}\right)$,
 so that $\cG(4)$ implies the constraint 
 $7x+12y+7z \le 72$,
 which is automatically satisfied. Finally, $\cG(4)$ yields $X(5)=\left(\tfrac{4147x+2592y+1037z}
 {5184},\tfrac{4147x+2592y+1037z}{5184},\tfrac{4147x+2592y+1037z}{5184},\tfrac{4147x+2592y+1037z}{5184}\right)$.

 In order to draw a picture of the feasible region for our parameters $x,y,z$ we additionally assume $x=z$, i.e.\ we restrict to mirror symmetric starting configurations. After plugging in and removing the redundant constraints $x+y>1$, $5x+4y>6$ we end up with $\left\{(x,y)\in\R_{\ge 0}^2\,:\, x\le 1,y\le 1,7x+6y>12\right\}$. In Figure~\ref{fig_feasible_region} we have filled the feasible region in blue and marked the edge that is not contained by a thicker line. We can see that also the conditions $x,y\ge 0$ can be removed.
\end{example}
\begin{figure}[htp]
 \begin{center}
  \begin{tikzpicture}[scale=5.8]
    \draw[->] (-0.1,0) -- (1.1,0) node[right] {$x$};
    \draw[->] (0,-0.1) -- (0,1.1) node[above] {$y$};
    \draw (-0.02,1) -- (0.02,1) node[near start, left] {$1$};
    \draw (1,-0.02) -- (1,0.02) node[near start,below] {$1$};

    \draw[very thick,blue] (1,0.8333) -- (0.8571,1) node[left,below] {$7x+6y>12\quad\quad\quad\quad\quad$};

    \filldraw[fill=blue!20,draw=blue!60]
        (1,1) -- (1,0.8333) -- (0.8571,1) -- cycle;

  \end{tikzpicture}
  \caption{Feasible region $\left\{(x,y)\in\R_{\ge 0}^2\,:\, x\le 1,y\le 1,7x+6y>12\right\}$ of Example~\ref{example_3}.}
  \label{fig_feasible_region}
  \end{center}
\end{figure}

\begin{lemma}
  For each $n\in\N_{>0}$ there exists a finite number of equivalence classes with respect to $\cong_{\cG}$.
\end{lemma}
\begin{proof}
  As in the proof of Lemma~\ref{lemma_finite_number_euqivalence_classes_1} we consider a sequence of length $F(n)+1$ of possible influence graphs $\cG(0),\dots,\cG(F(n))$. Now the entries of the starting configuration $X(0)$ are considered as variables and the presence or absence of edges in $\cG(t)$ is modeled as a linear constraint, see Example~\ref{example_4} and Section~\ref{sec_influence_graphs} for details. For each possibility there either exists at least one feasible solution or the feasible region is empty (which does not harm the upper bound).
 \end{proof}
We remark that the statement remains true for higher dimensional BC-processes.

\begin{example}\label{example_4}
  For $n=1$ agent there exists a unique equivalence class with respect to $\cong_{\cG}$. Here $x_1(0)\in\R$ and $\varepsilon\in\R_{\ge 0}$ are arbitrary and we have $X(t)=X(0)$ for all $t\in\N$, i.e.\ a freezing time of zero. For two agents we assume the ordering $x_1(0)\le x_2(0)$. If $x_2(0)-x_1(0)>\varepsilon$, then we have a freezing time of zero and no consensus. If $x_1(0)=x_2(0)$, then we also have a freezing time of zero but we also have consensus. If $x_1(0)<x_2(0)$ and $x_2(0)-x_1(0)\le \varepsilon$, then we have a freezing time of one and a consensus at $\tfrac{x_1(0)+x_2(0)}{2}$. Note that the two latter cases are graph equivalent, so that we will refine the notion of an influence graph in Section~\ref{sec_influence_graphs} in order to achieve that the freezing time is an invariant for the equivalence classes. In general we have a freezing time of zero for all cases where $\varepsilon=0$, so that we assume $\varepsilon>0$ in the following. For $n\ge 3$ agents the situation gets more and more involved. For $n=3$ we use affine transformations to normalize to $\varepsilon=1$, $x_1(0)=0$ and use the parameterization $X(0)=(0,x,x+y)$ (assuming $\varepsilon>0$).
\begin{itemize}
 \item If $x=0$ and $y=0$, then we have a freezing time of zero and consensus.\\[-10mm]
 \item If $x=0$ and $0<y\le 1$, then we have a freezing time of one and a consensus at $\tfrac{y}{3}$.\\[-10mm]
 \item If $x=0$ and $y>1$, then we have a freezing time of zero and fragmentation.\\[-10mm]
 \item If $0<x\le 1$ and $y=0$, then we have a freezing time of one and a consensus at $\tfrac{2x}{3}$.\\[-10mm]
 \item If $x>1$ and $y=0$, then we have a freezing time of zero and fragmentation.\\[-10mm]
 \item If $x>1$ and $y>1$, then we have a freezing time of zero and fragmentation.\\[-10mm]
 \item If $0<x\le 1$ and $y>1$, then we have a freezing time of one and fragmentation.\\[-10mm]
 \item If $x>1$ and $0<y\le 1$, then we have a freezing time of one and fragmentation.\\[-10mm]
 \item If $0<x\le 1$, $0<y\le 1$, and $x+y\le 1$, then we have a freezing time of one and a consensus at $\tfrac{2x+y}{3}$.\\[-10mm]
 \item If $0<x\le 1$, $0<y\le 1$, and $x+y> 1$, then we have $X(1)=\left(\tfrac{x}{2},\tfrac{2x+y}{3},x+\tfrac{y}{2}\right)$. Since $\tfrac{x+y}{2}\le 1$ we conclude a freezing time of $2$ and a consensus at $\tfrac{13x+5y}{18}$.\\[-10mm]
\end{itemize}
This small discussion shows $F(1)=0$, $F(2)=1$, and $F(3)=2$.
\end{example}

In Section~\ref{sec_influence_graphs} we will state a precise linear programming formulation for the problem whether there exists a starting configuration $X(0)$  that attains a given sequence of influence graphs. In the subsequent sections we will use these methods to gather some insight into properties and counterintuitive phenomena of BC-processes.

Besides the announced studies there is another reason to be interested in graph equivalence. In many studies it has been noted that BC-processes are highly vulnerable to small numerical deviations. When dealing with practical applications, where the values $x_i(0)$ of the starting configurations as well as the confidence level $\varepsilon$ may have to be determined by experiments or estimations, it seems more reasonable that we do not know the desired data points exactly but we can state intervals.
\begin{example}\label{example_5}
  Consider a BC-process for four agents where we only know $x_i(0)\in[i+0.9,i+1.1]$ for $1\le i\le 4$ and $\varepsilon\in[0.9,1.1]$. What can happen? Can we make any meaningful statements on the freezing time, the final width, or the degree of fragmentation? Using the computational methods described in the subsequent sections we state the freezing time can be any number, not equal to $3$, between $0$ and $5$, the number of connected components can be any number between $1$ and $4$, the minimum final width can be zero (for a consensus), and the maximum final width can be $3.2$ (for full fragmentation in the starting profile). So, essentially almost everything can happen. This is different for larger numbers of agents or tighter constraints on the parameters.
\end{example}

\section{Influence graphs and linear programming}\label{sec_influence_graphs}
Many properties of a BC-process can be read off from its sequence $\left(\cG(t)\right)_{t\in\N}$ of influence graphs. So, here we introduce a few basic graph theory concepts and refer the interested reader to literally any textbook on graph theory for more details. First of all, a \emph{graph} $G=(V,E)$ is given by a set $V$, whose elements are called \emph{vertices} or nodes, and a set $E\subseteq \big\{\{a,b\}\,:\, a,b\in V,a\neq b\big\}$, whose elements are called \emph{edges}. For our influence graphs $\cG(t)$ the set of vertices is given by the set $A$ of agents. There is an edge between agent $i\in A$ and agent $j\in A\backslash\{i\}$ at time $t$ iff the opinions are within the level of confidence, i.e.\ iff $\left|x_i(t)-x_j(t)\right|\le \varepsilon$. Since for a graph $G=(V,E)$ with $n$ vertices each of the ${n\choose 2}$ pairs of different vertices can either be contained in $E$ or not contained in $E$, we have $2^{n\choose 2}$ possible (labeled) graphs. We depict all possible (labeled) graphs in Figure~\ref{fig_labeled_graphs_3}

\begin{figure}
  \begin{center}
  $\cG_1$: \begin{tikzpicture}[node/.style={circle, draw, inner sep=2pt}]
    \node[node,label=below:$1$] (v1) at (0,0) {};
    \node[node,label=below:$2$] (v2) at (1,0) {};
    \node[node,label=below:$3$] (v3) at (2,0) {};
  \end{tikzpicture} \quad\quad\quad\quad
  $\cG_2$: \begin{tikzpicture}[node/.style={circle, draw, inner sep=2pt}]
    \node[node,label=below:$1$] (v1) at (0,0) {};
    \node[node,label=below:$2$] (v2) at (1,0) {};
    \node[node,label=below:$3$] (v3) at (2,0) {};
    \draw (v1) -- (v2);
  \end{tikzpicture} \quad\quad\quad\quad
  $\cG_3$: \begin{tikzpicture}[node/.style={circle, draw, inner sep=2pt}]
    \node[node,label=below:$1$] (v1) at (0,0) {};
    \node[node,label=below:$2$] (v2) at (1,0) {};
    \node[node,label=below:$3$] (v3) at (2,0) {};
    \draw (v2) -- (v3);
  \end{tikzpicture}\\[5mm]
  $\cG_4$: \begin{tikzpicture}[node/.style={circle, draw, inner sep=2pt}]
    \node[node,label=below:$1$] (v1) at (0,0) {};
    \node[node,label=below:$2$] (v2) at (1,0) {};
    \node[node,label=below:$3$] (v3) at (2,0) {};
    \draw[bend left=45] (v1) to (v3);
  \end{tikzpicture} \quad\quad\quad\quad
  $\cG_5$: \begin{tikzpicture}[node/.style={circle, draw, inner sep=2pt}]
    \node[node,label=below:$1$] (v1) at (0,0) {};
    \node[node,label=below:$2$] (v2) at (1,0) {};
    \node[node,label=below:$3$] (v3) at (2,0) {};
    \draw (v1) -- (v2);
    \draw[bend left=45] (v1) to (v3);
  \end{tikzpicture} \quad\quad\quad\quad
  $\cG_6$: \begin{tikzpicture}[node/.style={circle, draw, inner sep=2pt}]
    \node[node,label=below:$1$] (v1) at (0,0) {};
    \node[node,label=below:$2$] (v2) at (1,0) {};
    \node[node,label=below:$3$] (v3) at (2,0) {};
    \draw (v2) -- (v3);
    \draw[bend left=45] (v1) to (v3);
  \end{tikzpicture}\\[5mm]
  $\cG_7$: \begin{tikzpicture}[node/.style={circle, draw, inner sep=2pt}]
    \node[node,label=below:$1$] (v1) at (0,0) {};
    \node[node,label=below:$2$] (v2) at (1,0) {};
    \node[node,label=below:$3$] (v3) at (2,0) {};
    \draw (v1) -- (v2) -- (v3);
  \end{tikzpicture} \quad\quad\quad\quad
  $\cG_8$: \begin{tikzpicture}[node/.style={circle, draw, inner sep=2pt}]
    \node[node,label=below:$1$] (v1) at (0,0) {};
    \node[node,label=below:$2$] (v2) at (1,0) {};
    \node[node,label=below:$3$] (v3) at (2,0) {};
    \draw (v1) -- (v2) -- (v3);
    \draw[bend left=45] (v1) to (v3);
  \end{tikzpicture}
    \caption{Labeled graphs with three vertices.}
    \label{fig_labeled_graphs_3}
  \end{center}
\end{figure}
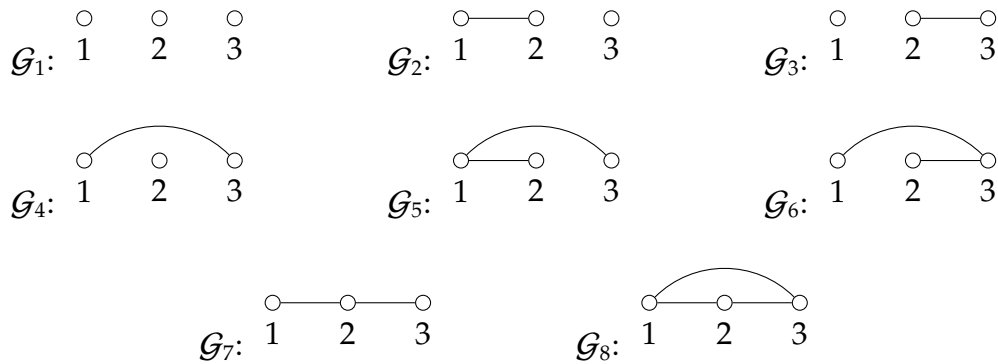

Removing the labels, i.e.\ the names of the vertices, we speak of unlabeled graphs. We say that two graphs $G=(V,E)$ and $G'=(V',E')$ are isomorphic if there exists a bijection $\pi\colon V\to V'$ such that $\{a,b\}\in E$ iff
$\{\pi(a),\pi(b)\}\in E'$ for all $a,b\in V$ with $a\neq b$. So, the graphs $\cG_2$, $\cG_3$, $\cG_4$ as well as the graphs $\cG_5$, $\cG_6$, $\cG_7$ are pairwise isomorphic, i.e.\ there exist exactly four non-isomorphic (unlabeled) graphs with three vertices. At this point we remark that unlabeled graphs can be useful to reduce the number of possible cases but there are also some subtleties that we have to deal with when using them in our context. As an example we mention that the graph $\cG_4$ does not occur in our discussion on BC-processes for three agents in Example~\ref{example_4} since it is excluded by our specific parameterization. More precisely, the edge $\{1,3\}$ relates to $x+y\le \varepsilon=1$, the non-edge $\{1,2\}$ to $x>\varepsilon=1$, and the non-edge $\{2,3\}$ to $y>\varepsilon=1$, which is an inconsistent linear inequality system, i.e.\ it does not admit a feasible solution. Of course, we can always stick to labeled graphs to avoid these difficulties with parameterizations or assumed orderings like $x_1(0)\le x_2(0)\le \dots\le x_n(0)$.

A \emph{complete graph} $G=(V,E)$ is a graph where $E$ contains all $2$-subsets of $V$, i.e.\ the complete graph for three vertices is $\cG_8$. We remark that if $\cG(t)$ is a complete graph for a given BC-process $\cP$, then we either have that all entries of $X(t)$ are equal, i.e.\ a consensus, or this property holds for $X(t+1)$.
So, we can detect a consensus more or less directly from the influence graphs. In order to describe fragmentation, we introduce more graph-theoretic notation. Given a graph $G=(V,E)$, a \emph{path} between $U_0\in V$ and $U_l\in V$ is a sequence of vertices $\left(U_0,U_1,\dots,U_l\right)$ with $U_i\in V$ such that $\left\{U_i,U_{i+1}\right\}\in E$ for all $0\le i< l$. In other words, there is a path between two vertices if there exists a sequence of edges connecting them.
The \emph{connected component} $\cC(U)$ of a vertex $U\in V$ is the set of vertices $U'\in V$ that admit a path between $U$ and $U'$ (including $U$ itself). Being contained in a connected component is an equivalence relation, i.e.\ the set of vertices is partitioned into disjoint connected components. E.g.\ the graph $\cG_4$ consists of the two connected components $\{1,3\}$ and $\{2\}$. In our context we write $\cC_i(t)$ for the connected component of agent $i\in A$ and time $t\in \N$. The number of connected components is a measurement for fragmentation. Clearly, there can be at most $n$ connected components and this maximum can always be attained in our context if the initial opinions $x_i(0)$ are pairwise different and the confidence level $\varepsilon$ is sufficiently small (e.g.\ $\varepsilon=0$). For $n=3$ graph $\cG_1$ has three connected components. If a given graph $G=(V,E)$ has a unique connected component, which then has to equal $V$, then we say that the graph is \emph{connected} and \emph{disconnected} otherwise. For three vertices the connected graphs are given by $\cG_5$, $\cG_6$, $\cG_7$, and $\cG_8$.
Following \cite{krause2000discrete}, one speaks of an $\varepsilon$-profile if the influence graph $\cG(t)$ is connected, see also \cite{hegselmann2023bounded}; in \cite{dittmer2001consensus} the term $\varepsilon$-chain is used.

Let us now face the problem to decide whether a given graph can be the influence graph of a given BC-process. The corresponding analytical description can be stated easily:
\begin{lemma}
  Let $\cP=(X(0),\varepsilon)$ be a BC-process for $n$ agents and $G=(V,E)$ be a graph with vertex set $\{1,\dots,n\}$. Then $G$ is the influence graph of $\cP$ at time $t\in\N$ iff $\left|x_i(t)-x_j(t)\right|\le\varepsilon$ for all $\{i,j\}\in E$ and $\left|x_i(t)-x_j(t)\right|>\varepsilon$ for all $\{i,j\}\notin E$ with $1\le i<j\le n$.
\end{lemma}
Of course we did nothing more than spelling out the conditions for $G=\cG(t)$ according to the definition of the influence graph at time $t$. However this is useful nevertheless if we introduce variables and assume $x_1(0)\le x_2(0)\le \dots \le x_n(0)$.
\begin{lemma}
  A (suitably labeled) graph $G(V,E)$ can be the influence graph $\cG(0)$ of a BC-process for $n$ agents iff there exist $x_1,\dots,x_n\in\R$ and $\varepsilon\in\R_{\ge 0}$ satisfying $x_1\le x_2\le\dots\le x_n$, $ x_i-x_j\le\varepsilon$ for all $\{i,j\}\in E$ with $i>j$, and $x_i-x_j>\varepsilon$ for all $\{i,j\}\notin E$ with $n\ge i>j\ge 1$.
\end{lemma}
So, we have transferred the problem to the question whether a linear inequality system admits a feasible solution, i.e.\ a choice of values for all variables
such that all constraints are satisfied. Strict inequalities are a mess for numerics as well as for mathematical properties, so that we further reformulate to an optimization problem, more precisely a linear program (LP).
\begin{lemma}
  \label{lemma_lp_realization_check}
  A (suitably labeled) graph $G(V,E)$ can be the influence graph $\cG(0)$ of a BC-process for $n$ agents iff the linear program
  $$
  \begin{aligned}
  \text{maximize} \quad & \Delta \\
  \text{subject to}\quad \\
    x_i-x_j -\varepsilon & \le 0\quad\forall \{i,j\}\in E\text{ with } i>j \\
    x_i-x_j -\varepsilon-\Delta & \ge 0\quad\forall \{i,j\}\notin E\text{ with } 1\le j<i\le n \\
    x_i-x_{i+1} &\le 0\quad\forall 1\le i\le n-1\\\
    \varepsilon &\ge 0
  \end{aligned}
  $$
  admits a solution with target value strictly larger than $0$.
\end{lemma}
We observe that our linear program is feasible since $x_1=\dots = x_n=\varepsilon=\Delta=0$ satisfies all constraints.
If we want to ensure that our linear program is bounded, then we can utilize affine transformations and additionally require $x_1=0$, $\varepsilon=1$, $x_n\le n$, $\Delta\le 1$. Here e.g.\ $x_1=\dots=x_n=0$, $\varepsilon=1$, $\Delta=-1$ is a feasible solution.
Since all coefficients are rational numbers e.g.\ the simplex algorithm can be evaluated with exact arithmetic and compute a rational solution $\left(x^\star,\Delta^\star,\varepsilon^\star\right)\in \mathbb{Q}_{\ge 0}^{n}\,\times\,\mathbb{Q}\,\times\, \mathbb{Q}_{\ge 0}$ or decide that no such solution exists. Clearly we can decide exactly whether the computed $\Delta^\star\in\mathbb{Q}$ is larger than zero or not.

\begin{example}\label{example_6}
   Let $\cG_1$ be the graph with four vertices given by the path $(1,2,3,4)$, i.e.\ $\cG_1$ equals $\cG(0)=\cG(1)$ in Example~\ref{example_2}. As $\cG_2$ we choose a complete graph with four vertices, i.e.\ $\cG_2$ equals $\cG(3)=\cG(4)$ in Example~\ref{example_2}. Now we want to computationally check whether a BC-process can start with $\cG_1$ as influence graph and then transition to $\cG_2$ as influence graph in the next step. Here we can utilize Lemma~\ref{lemma_lp_realization_check} as follows. First we use it for $\cG_1$ with variables $\left(x_1^1,\dots,x_4^1,\varepsilon,\Delta\right)$ and for $\cG_2$ with variables $\left(x_1^2,\dots,x_4^2,\varepsilon,\Delta\right)$ adding the normalization $x_1^1=0$, $x_4^1\le 4$, $\varepsilon=1$. Then we relate $X^1=\left(x_1^1,x_2^1,x_3^1,x_4^1\right)$ and $X^2=\left(x_1^2,x_2^2,x_3^2,x_4^2\right)$ via the updating rule (\ref{eq_update_rule}):
   \begin{align}
   \text{maximize} \quad & \Delta \\
   \text{subject to}\quad \nonumber\\
   x_2^1-x_1^1-\varepsilon &\le 0 \\
   x_3^1-x_2^1-\varepsilon &\le 0 \\
   x_4^1-x_3^1-\varepsilon &\le 0 \\
   x_3^1-x_1^1-\varepsilon-\Delta &\ge 0 \\
   x_4^1-x_1^1-\varepsilon-\Delta &\ge 0 \\
   x_4^1-x_2^1-\varepsilon-\Delta &\ge 0 \\
   x_1^1-x_2^1 &\le 0 \\
   x_2^1-x_3^1 &\le 0 \\
   x_3^1-x_4^1 &\le 0 \\
   x_1^1 &=0\\
   x_4^1 &\le 4\\
   \varepsilon &=1
   \end{align}\pagebreak
   \begin{align}
   x_2^2-x_1^2-\varepsilon &\le 0 \label{ie_lp_first}\\
   x_3^2-x_1^2-\varepsilon &\le 0 \\
   x_4^2-x_1^2-\varepsilon &\le 0 \label{ie_lp_essential}\\
   x_3^2-x_2^2-\varepsilon &\le 0 \\
   x_4^2-x_2^2-\varepsilon &\le 0 \\
   x_4^2-x_3^2-\varepsilon &\le 0 \label{ie_lp_last}\\
   x_1^2-x_2^2 &\le 0 \\
   x_2^2-x_3^2 &\le 0 \\
   x_3^2-x_4^2 &\le 0 \\
   \tfrac{1}{2} x_1^1+\tfrac{1}{2} x_2^1 -x_1^2 & = 0\\
   \tfrac{1}{3} x_1^1+\tfrac{1}{3} x_2^1+\tfrac{1}{3} x_3^1 -x_2^2 & = 0\\
   \tfrac{1}{3} x_2^1+\tfrac{1}{3} x_3^1+\tfrac{1}{3} x_4^1 -x_3^2 & = 0\\
   \tfrac{1}{2} x_3^1+\tfrac{1}{2} x_4^1 -x_4^2 & = 0\\
    \varepsilon &\ge 0
  \end{align}
  Using an LP solver we can compute an optimal solution given by $X^1=(0,1,1,2)$, $X^2=\left(\tfrac{1}{2},\tfrac{2}{3},\tfrac{4}{3},\tfrac{3}{2}\right)$, $\varepsilon=1$, and $\Delta=0$, so that the direct transition from $\cG_1$ to $\cG_2$ (as subsequent influence graphs) is impossible.
  We have coinciding intermediate influence graphs $\cG(2)$ in Example~\ref{example_2} and $\cG(3)$ in Example~\ref{example_1}.
\end{example}

\medskip

Of course we can use Lemma~\ref{lemma_lp_realization_check}, as indicated in Example~\ref{example_6}, to computationally check whether any given sequence of graphs $\cG_1,\dots,\cG_l$ can be realized as a subsequent influence graphs of a BC-process. As we see in Example~\ref{example_6} a rather large number of constraints is generated. However, some of them are redundant and can be dropped without changing the feasible region. We have already written down the constraint $\varepsilon\ge 0$ for $\cG_1$ and $\cG_2$, as required by Lemma~\ref{lemma_lp_realization_check}, only once. Since we use the normalization $\varepsilon=1$ anyway, we can drop it. Due to the inheritance of an initial ordering of the opinions, see Lemma~\ref{lemma_basic}.(1), incorporating the updating rule~(\ref{eq_update_rule}) allows us to remove the ordering constraints $x_i-x_{i+1}\le 0$ for all but the first graph. There is another hidden dominance relation. E.g.\ we can replace the Inequalities (\ref{ie_lp_first})-(\ref{ie_lp_last}) in the LP of Example~\ref{example_6} just by Inequality~(\ref{ie_lp_essential}). In order to describe the general phenomenon we introduce more notation.
\begin{definition}
  Let $\cP=(X(0),\varepsilon)$ be a BC-process for $n$ agents with $x_1(0)\le\dots\le x_n(0)$. By $l_i(t)$ we denote the smallest integer $j\ge 1$ such that $x_i(t)-x_j(t)\le \varepsilon$ and by $r_i(t)$ the largest integer $h\le n$ such that $x_h(t)-x_i(t)\le \varepsilon$, where $t\in \N$ and $1\le i\le n$.
\end{definition}
In words, $l_i(t)$ is the leftmost agent (having the smallest possible index if equality occurs) that still influences agent~$i$ at time $t$. Similarly, $r_i(t)$ is the rightmost agent (having the largest possible index if equality occurs) that still influences agent~$i$ at time $t$. The assumed ordering $x_1(0)\le\dots\le x_n(0)$ and Lemma~\ref{lemma_basic}.(1) directly implies:
\begin{lemma}\label{lemma_insider_interval}
If $\cP=(X(0),\varepsilon)$ is a BC-process with $n$ agents satisfying $x_1(0)\le \dots\le x_n(0)$, then we have
\begin{equation}
  I_i(t)=\left\{l_i(t),l_{i}(t)+1,\dots,r_i(t)\right\}
\end{equation}
for all $i\in A$ and $t\in\N$.
Moreover, $l_1(t)\le l_2(t)\le\dots\le l_n(t)$ and $r_1(t)\le r_2(t)\le\dots\le r_n(t)$ for all $t\in\mathbb N$.
\end{lemma}
\begin{definition}
   Let $\cG=(V,E)$ be a graph with vertex set $\{1,\dots,n\}$. By $l(i)$ we denote the smallest integer $1\le j<i$ such that $\{j,i\}\in E$ and set $l(i)=i+1$ if no such $j$ exists, where $1\le i\le n$. Similarly, by $r(i)$ we denote the largest integer $i<h\le n$ such that $\{i,h\}\in E$ and we set $r(i)=i-1$ if no such integer $h$ exists, where $1\le i\le n$.
\end{definition}
With this, Lemma~\ref{lemma_insider_interval} directly implies:
\begin{lemma}\label{lemma_edge_intervals}
  Let $\cG=(V,E)$ be a graph with vertex set $V=\{1,\dots,n\}$. If there exists a BC-process $\cP=(X(0),\varepsilon)$ with $n$ agents satisfying $x_1(0)\le \dots\le x_n(0)$ and $\cG=\cG(t)$ for some $t\in\N$, then we have
  \begin{equation}\label{eq_window}
    E=\bigcup_{1\le i\le n} \big\{\{i,j\}\,:\, l(i)\le j\le r(i), j\neq i\big\}.
  \end{equation}
\end{lemma}

\begin{lemma}
  In Lemma~\ref{lemma_lp_realization_check} we can replace the constraints $x_i-x_j-\varepsilon\le 0$ for all $\{i,j\}\in E$ with $i>j$ by $x_i-x_j-\varepsilon\le 0$ for all $\{i,j\}\in E'$, where
  \begin{equation}
    E'=\left(\big\{1,r(1)\big\}\cup\big\{\{i,r(i)\}\,:\,2\le i\le n,r(i)>r(i-1),r(i)>i\big\}\right) \backslash \big\{\{0,1\}\big\}.
  \end{equation}
\end{lemma}

\begin{lemma}
   Let $\cG$ satisfy Equation~(\ref{eq_window}).
   In Lemma~\ref{lemma_lp_realization_check} we can replace the constraints $x_i-x_j-\varepsilon-\Delta\ge 0$ for all $\{i,j\}\notin E$ with $n\ge i>j\ge 1$ by $x_i-x_j-\varepsilon-\Delta\ge 0$ for all $\{i,j\}\in \widetilde{E}$ with $i>j$, where
   \begin{equation}
     \widetilde{E}=\big\{\{\min\{l(i),i\}-1,i\}\,:2\le i\le n, 1<l(i),l(i)>l(i-1) \big\}.
   \end{equation}
\end{lemma}

\begin{example}\label{example_7}
  If we remove the superfluous constraints in the LP from Example~\ref{example_6} as described above, we obtain the following LP.
  \begin{align}
   \text{maximize} \quad & \Delta \\
   \text{subject to}\quad \nonumber\\
   x_2^1-x_1^1-\varepsilon &\le 0 \\
   x_3^1-x_2^1-\varepsilon &\le 0 \\
   x_4^1-x_3^1-\varepsilon &\le 0 \\
   x_3^1-x_1^1-\varepsilon-\Delta &\ge 0 \label{ie_dual_1}\\
   x_4^1-x_2^1-\varepsilon-\Delta &\ge 0 \label{ie_dual_2}\\
   x_1^1-x_2^1 &\le 0 \\
   x_2^1-x_3^1 &\le 0 \\
   x_3^1-x_4^1 &\le 0 \\
   x_1^1 &=0\\
   x_4^1 &\le 4\\
   \varepsilon &=1\\
   x_4^2-x_1^2-\varepsilon &\le 0 \label{ie_dual_3}\\
   \tfrac{1}{2} x_1^1+\tfrac{1}{2} x_2^1 -x_1^2 & = 0\label{ie_dual_4}\\
   \tfrac{1}{3} x_1^1+\tfrac{1}{3} x_2^1+\tfrac{1}{3} x_3^1 -x_2^2 & = 0\\
   \tfrac{1}{3} x_2^1+\tfrac{1}{3} x_3^1+\tfrac{1}{3} x_4^1 -x_3^2 & = 0\\
   \tfrac{1}{2} x_3^1+\tfrac{1}{2} x_4^1 -x_4^2 & = 0\label{ie_dual_5}
  \end{align}
  An optimal solution is again given by $X^1=(0,1,1,2)$, $X^2=\left(\tfrac{1}{2},\tfrac{2}{3},\tfrac{4}{3},\tfrac{3}{2}\right)$, $\varepsilon=1$, and $\Delta=0$. The sum of $-\tfrac{1}{2}$ times Inequality~(\ref{ie_dual_1}), $-\tfrac{1}{2}$ times Inequality~(\ref{ie_dual_2}), $1$ times Inequality~(\ref{ie_dual_3}), $-1$ times Inequality~(\ref{ie_dual_4}), and $1$ times Inequality~(\ref{ie_dual_5}) gives $\Delta\le 0$, which proves that the optimal target value for $\Delta$ is indeed $0$. We remark that such multipliers always exist if the LP is feasible and bounded, which it is when we use the stated normalization, and refer to the strong duality theorem of linear optimization (contained in any textbook on linear optimization).
\end{example}

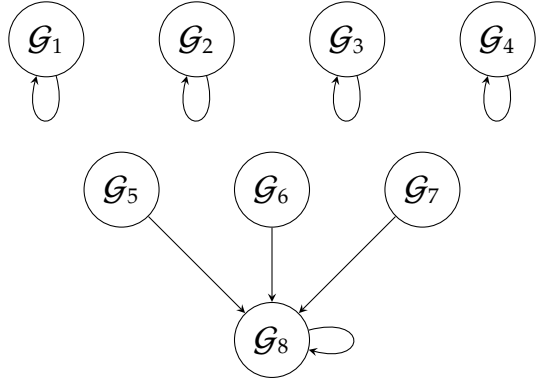
\begin{figure}[htp]
  \begin{center}
  \begin{tikzpicture}[>=stealth, node distance=2cm,
    vertex/.style={circle,draw,minimum size=6pt}]

    \node[vertex] (A1) at (0,4) {$\cG_1$};
    \node[vertex] (A2) at (2,4) {$\cG_2$};
    \node[vertex] (A3) at (4,4) {$\cG_3$};
    \node[vertex] (A4) at (6,4) {$\cG_4$};
    \node[vertex] (A5) at (1,2) {$\cG_5$};
    \node[vertex] (A6) at (3,2) {$\cG_6$};
    \node[vertex] (A7) at (5,2) {$\cG_7$};
    \node[vertex] (A8) at (3,0) {$\cG_8$};

    \draw[->] (A5) --  (A8);
    \draw[->] (A6) --  (A8);
    \draw[->] (A7) --  (A8);

    \draw[->] (A1) to[loop below]  ();
    \draw[->] (A2) to[loop below] ();
    \draw[->] (A3) to[loop below]  ();
    \draw[->] (A4) to[loop below]  ();
    \draw[->] (A8) to[loop right] ();
  \end{tikzpicture}
  \caption{Transition graph for graphs with three vertices, see Figure~\ref{fig_labeled_graphs_3} for the labeling.}
  \label{fig_transition_3}
  \end{center}
\end{figure}

So, we have a machinery at hand to compute the \emph{transition graph} of the possible influence graphs of BC-processes with a given finite number of agents, i.e.\ a \emph{directed graph}\footnote{A \emph{directed graph} $G=(V,E)$ consists of a set of vertices $V$ and a set of edges $E\subseteq V\times V$. For an element $(i,j)\in E$ we speak of an edge directed from $i$ to $j$.} with \emph{loops}\footnote{A \emph{loop} in a directed graph $G=(V,E)$ is an edge $E \ni e=(i,i)$ which is directed from vertex $i\in V$ to itself.}  with the possible influence graphs as vertices and a directed edge from one to another whenever a transition between the two corresponding influence graphs is possible in principle. Using our enumeration of graphs with three vertices in Figure~\ref{fig_labeled_graphs_3} we draw the corresponding transition graph in Figure~\ref{fig_transition_3}. For these small cases we actually do not need our general machinery, see Example~\ref{example_4} for the details of a down-to-earth computation and case analysis. For a larger number of agents it nevertheless becomes indispensable. From Figure~\ref{fig_transition_3} we can directly read off that the maximum freezing time $F(3)$ for $3$ agents is at most $2$ noting that the loop for $\cG_8$ ends in a consensus after one step. Here indeed all possible paths in the transition graph can be realized, which is different for $n\ge 4$ agents.

We remark that at and after freezing time the influence graph of a BC-process has to be one with a loop in the transition graph. For $n=3$ vertices all possible graphs can be realized by a BC-process (ignoring the initial ordering of the opinions) and all vertices with loops in the transition graph can indeed be attained in the stable state. We will see that both properties will not hold for $n=4$ agents.

\begin{figure}[htp]
  \begin{center}
    $\cG_1$:  \begin{tikzpicture}[node/.style={circle, draw, inner sep=2pt}]
    \node[node,label=below:$1$] (v1) at (0,0) {};
    \node[node,label=below:$2$] (v2) at (1,0) {};
    \node[node,label=below:$3$] (v3) at (2,0) {};
    \node[node,label=below:$4$] (v4) at (3,0) {};
    \draw (v1) -- (v2);
    \draw[bend left=45] (v1) to (v3);
    \draw[bend left=45] (v1) to (v4);
  \end{tikzpicture}\quad\quad
  $\cG_2$:  \begin{tikzpicture}[node/.style={circle, draw, inner sep=2pt}]
    \node[node,label=below:$1$] (v1) at (0,0) {};
    \node[node,label=below:$2$] (v2) at (1,0) {};
    \node[node,label=below:$3$] (v3) at (2,0) {};
    \node[node,label=below:$4$] (v4) at (3,0) {};
    \draw (v1) -- (v2)--(v3);
    \draw[bend left=45] (v2) to (v4);
  \end{tikzpicture}\\[5mm]
  $\cG_3$:  \begin{tikzpicture}[node/.style={circle, draw, inner sep=2pt}]
    \node[node,label=below:$1$] (v1) at (0,0) {};
    \node[node,label=below:$2$] (v2) at (1,0) {};
    \node[node,label=below:$3$] (v3) at (2,0) {};
    \node[node,label=below:$4$] (v4) at (3,0) {};
    \draw (v2) -- (v3)--(v4);
    \draw[bend left=45] (v1) to (v3);
  \end{tikzpicture}\quad\quad
  $\cG_4$:  \begin{tikzpicture}[node/.style={circle, draw, inner sep=2pt}]
    \node[node,label=below:$1$] (v1) at (0,0) {};
    \node[node,label=below:$2$] (v2) at (1,0) {};
    \node[node,label=below:$3$] (v3) at (2,0) {};
    \node[node,label=below:$4$] (v4) at (3,0) {};
    \draw (v3) -- (v4);
    \draw[bend left=45] (v1) to (v4);
    \draw[bend left=45] (v2) to (v4);
  \end{tikzpicture}
  \caption{Stars with four vertices.}
  \label{fig_star_4}
  \end{center}
\end{figure}
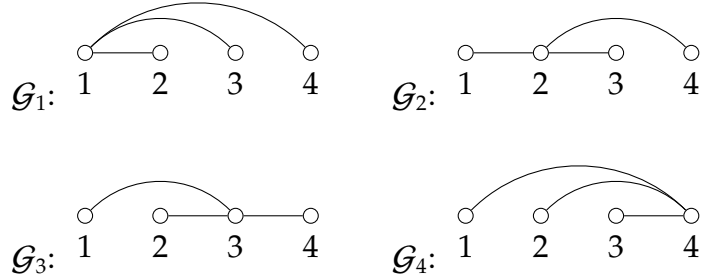

A \emph{star} is a graph with one central vertex that has an edge to every other vertex and there are no further edges. In Figure~\ref{fig_star_4} we have depicted all labeled stars with four vertices. We remark that we can directly use Lemma~\ref{lemma_lp_realization_check} to computationally verify that none of them can be an influence graph of a BC-process.\footnote{Of course for every $n\ge 3$ there exists a unique unlabeled star with $n$ vertices admitting a more symmetric drawing with the center in the middle of a unit circle and the other vertices distributed equally along the boundary of the unit circle. Since already a star with four vertices cannot be realized as the influence graph of a BC-process, no star with $n\ge 4$ vertices can be realized as an influence graph of a BC-process.}

\begin{definition}\label{def_interval_graph}
  A graph $G=(V,E)$ with vertex set $V=\{1,\dots,n\}$ is an \emph{interval graph} if there exists intervals $I_1,\dots,I_n\subseteq\R$ such that $\{i,j\}\in E$ iff $I_i\cap I_j\neq \emptyset$ for $i\neq j$. If, moreover, all intervals can be chosen of the same length, then $G$ is called a \emph{unit interval graph}.
\end{definition}

\begin{table}[htp]
  \begin{center}\footnotesize
  \begin{tabular}{rrrrrrr}
  \hline
  n & $A000088(n)$ & $A005975(n)$ & $A005217(n)$ & $A001349(n)$ & $A005976(n)$ & $A007123(n)$ \\
  \hline
  1  &        1 &     1 &    1 &        1 &     1 &    1 \\
  2  &        2 &     2 &    2 &        1 &     1 &    1 \\
  3  &        4 &     4 &    4 &        2 &     2 &    2 \\
  4  &       11 &    10 &    9 &        6 &     5 &    4 \\
  5  &       34 &    27 &   21 &       21 &    15 &   10 \\
  6  &      156 &    92 &   55 &      112 &    56 &   26 \\
  7  &     1044 &   369 &  151 &      853 &   250 &   76 \\
  8  &    12346 &  1807 &  447 &    11117 &  1328 &  232 \\
  9  &   274668 & 10344 & 1389 &   261080 &  8069 &  750 \\
  10 & 12005168 & 67659 & 4502 & 11716571 & 54962 & 2494 \\
  11 & 1018997864 & 491347 & 15046 & 1006700565 & 410330 & 8524 \\
  12 & 165091172592 & 3894446 & 51505 & 164059830476 & 3317302 & 29624 \\
  13 & 50502031367952 & 33278992 & 179463 & 50335907869219 & 28774874 & 104468 \\
  14 & 29054155657235488 & 304256984 & 634086 & 29003487462848061 & 266242936 & 372308 \\
  15 & 31426485969804308768 & 2960093835 & 2265014 & 31397381142761241960 &2616100423 & 1338936 \\
  \hline
  \end{tabular}
  \caption{Number of unlabeled graphs for different graph classes. See Table~\ref{table_captions} for an explanation of the captions.}
  \label{table_number_of_graphs}
  \end{center}
\end{table}

For influence graphs $\cG(t)$ those intervals in Definition~\ref{def_interval_graph} are naturally given by $\left[x_i(t),x_i(t)+\varepsilon\right]$ -- all having the same length $\varepsilon$. Characterizations of interval graphs can e.g.\ be found in \cite{golumbic1980algorithmic}. An interval graph is a unit interval graph iff it does not contain a star with four vertices as an induced subgraph \cite{roberts1969indifference}. Interval graphs as well as unit interval graphs can be easily counted, see e.g.\ \cite{hanlon1982counting} for generating functions.\footnote{A \emph{generating function} is a formal power series that encodes an infinite sequence of numbers $\left(a_n\right)_{n\in\N}$ as the coefficients of the series $a(x)=\sum_{n\in\N} a_nx^n$.} Our observation in Lemma~\ref{lemma_edge_intervals} is close to the characterization of unit interval graphs in \cite[Theorem 4]{olariu1991optimal}: A graph $G=(V,E)$ is a unit interval graph iff there exists a linear order $\prec$ on $V$ such that for every choice of vertices $u$, $v$, $w$ with $u\prec v\prec w$, $\{u,w\}\in E$ implies $\{u,v\},\{v,w\}\in E$. With this, the property of being a (unit) interval graph can be checked in $O(\#V+\#E)$ time \cite{olariu1991optimal}.

\begin{table}[htp]
  \begin{center}
  \begin{tabular}{rl}
  \hline
  $A000088(n)$ & graphs with $n$ vertices                         \\
  $A005975(n)$ & interval graphs with $n$ vertices                \\
  $A005217(n)$ & unit interval graphs with $n$ vertices           \\
  $A001349(n)$ & connected graphs with $n$ vertices               \\
  $A005976(n)$ & connected interval graphs with $n$ vertices      \\
  $A007123(n)$ & connected unit interval graphs with $n$ vertices \\
  \hline
  \end{tabular}
  \caption{Captions in Table~\ref{table_number_of_graphs} -- sequences in the OEIS.}
  \label{table_captions}
  \end{center}
\end{table}

In Table~\ref{table_number_of_graphs} we have listed the number of unlabeled graphs for up to fifteen vertices and those with the additional properties of being an interval or unit interval graph as well as the property of being connected. As captions we have used the A-numbers in the {\lq\lq}On-Line Encyclopedia of Integer Sequences{\rq\rq} (OEIS)\footnote{OEIS Foundation Inc. (2026), The On-Line Encyclopedia of Integer Sequences, Published electronically at \url{https://oeis.org}.}, see Table~\ref{table_captions} for verbal descriptions. The corresponding numbers of labeled graphs are much larger and can also be found in the OEIS. The impact of restricting general graphs to unit interval graphs, as candidates for influence graphs, is quite tremendous for $n\ge 8$ agents. In some situations we can also restrict ourselves to connected unit interval graphs.

\begin{lemma}\label{lemma_decompose} (\cite[Proposition 2]{blondel2009krause}, \cite[proof of Theorem 2]{krause2000discrete})\\
Let $\cP=(X(0),\varepsilon)$ be a BC-process for $n$ agents with $x_1(0)\le\dots\le x_n(0)$. If at time $t\in\N$ and for some agent $1\le i<n$ we have $\left|x_i(t)-x_{i+1}(t)\right|>\varepsilon$, then we have $\left|x_i(t')-x_{i+1}(t')\right|>\varepsilon$ for all $t'\ge t$. So, the BC-process can be decomposed into two subprocesses, one consisting of agents $\{1, \ldots, i\}$ and the other of agents $\{i+1,\dots,n\}$, each evolving independently after time $t$.
\end{lemma}
We remark that the decomposition from Lemma~\ref{lemma_decompose} is only valid for a one-dimensional BC-process, to which we restrict in this paper. In higher dimensions, even for dimension two, the dynamics is much more complex.

If the influence graph $\cG(0)$ is connected, then the linear ordering $\prec$ associated to the interval graph allows us to consider unlabeled graphs. More precisely, the assumption $x_1(0)\le\dots\le x_n(0)$ almost uniquely determines the labeling of a given interval graph. The only degree of freedom left is to reverse the ordering $\prec$, i.e.\ to mirror the
BC-process at an arbitrary point, which is an affine transformation. Let us consider the labeled graphs with three vertices from Figure~\ref{fig_labeled_graphs_3} again. Assuming the ordering $x_1(0)\le x_2(0)\le x_3(0)$ the graphs $\cG_4$, $\cG_5$, and $\cG_6$ cannot be realized as a BC-process, which leaves the five graphs $\cG_1$, $\cG_2$, $\cG_3$, $\cG_7$, and $\cG_8$.  Going in line with Table~\ref{table_number_of_graphs}, graphs $\cG_2$ and $\cG_3$ are isomorphic. Corresponding BC-processes can indeed be obtained by mirroring. The other three graphs permit mirror symmetric BC-processes. The phenomenon also occurs for connected unit interval graphs with at least four vertices, see Figure~\ref{fig_mirror}, which are isomorphic as graphs.
\begin{figure}[htp]
\begin{center}
 \begin{tikzpicture}[node/.style={circle, draw, inner sep=2pt}]
    \node[node,label=below:$1$] (v1) at (0,0) {};
    \node[node,label=below:$2$] (v2) at (1,0) {};
    \node[node,label=below:$3$] (v3) at (2,0) {};
    \node[node,label=below:$4$] (v4) at (3,0) {};
    \draw (v1) -- (v2) -- (v3) -- (v4);
    \draw[bend left=45] (v1) to (v3);
  \end{tikzpicture}\quad\quad\quad\quad
  \begin{tikzpicture}[node/.style={circle, draw, inner sep=2pt}]
    \node[node,label=below:$1$] (v1) at (0,0) {};
    \node[node,label=below:$2$] (v2) at (1,0) {};
    \node[node,label=below:$3$] (v3) at (2,0) {};
    \node[node,label=below:$4$] (v4) at (3,0) {};
    \draw (v1) -- (v2) -- (v3) -- (v4);
    \draw[bend left=45] (v2) to (v4);
  \end{tikzpicture}
  \caption{Two isomorphic connected unit interval graphs with four vertices.}
  \label{fig_mirror}
\end{center}
\end{figure}
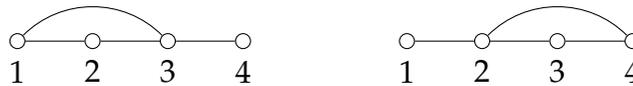

If we have several connected components, depending on the application, we may have to take the orderings of the connected components and their orientations into account. So, if we e.g.\ have two components as in Figure~\ref{fig_mirror} the question is whether it makes a difference if they are oriented in the same direction or in different directions. Using the decomposition from Lemma~\ref{lemma_decompose} we here do not take orderings of the connected components or orientations into account, i.e.\ we consider unlabeled unit interval graphs. So, e.g.\ we do not differentiate between graphs $\cG_2$ and $\cG_3$ in Figure~\ref{fig_labeled_graphs_3}. For $n=4$ agents we consider the nine unit interval graphs depicted in Figure~\ref{fig_unit_interval_graphs_4}. The corresponding transition graph is given in Figure~\ref{fig_transition_4}.

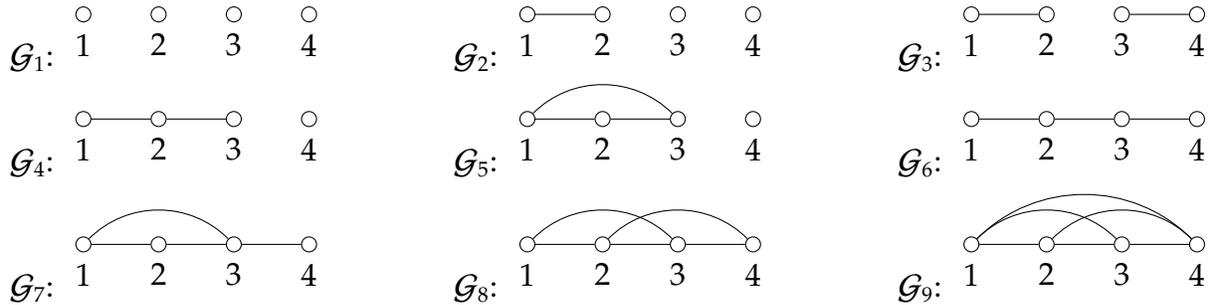
\begin{figure}[htp]
  \begin{center}
  $\cG_1$: \begin{tikzpicture}[node/.style={circle, draw, inner sep=2pt}]
    \node[node,label=below:$1$] (v1) at (0,0) {};
    \node[node,label=below:$2$] (v2) at (1,0) {};
    \node[node,label=below:$3$] (v3) at (2,0) {};
    \node[node,label=below:$4$] (v4) at (3,0) {};
  \end{tikzpicture}\quad\quad\quad\quad
  $\cG_2$: \begin{tikzpicture}[node/.style={circle, draw, inner sep=2pt}]
    \node[node,label=below:$1$] (v1) at (0,0) {};
    \node[node,label=below:$2$] (v2) at (1,0) {};
    \node[node,label=below:$3$] (v3) at (2,0) {};
    \node[node,label=below:$4$] (v4) at (3,0) {};
    \draw (v1) -- (v2);
  \end{tikzpicture}\quad\quad\quad\quad
  $\cG_3$: \begin{tikzpicture}[node/.style={circle, draw, inner sep=2pt}]
    \node[node,label=below:$1$] (v1) at (0,0) {};
    \node[node,label=below:$2$] (v2) at (1,0) {};
    \node[node,label=below:$3$] (v3) at (2,0) {};
    \node[node,label=below:$4$] (v4) at (3,0) {};
    \draw (v1) -- (v2);
    \draw (v3) -- (v4);
  \end{tikzpicture}\\
  $\cG_4$: \begin{tikzpicture}[node/.style={circle, draw, inner sep=2pt}]
    \node[node,label=below:$1$] (v1) at (0,0) {};
    \node[node,label=below:$2$] (v2) at (1,0) {};
    \node[node,label=below:$3$] (v3) at (2,0) {};
    \node[node,label=below:$4$] (v4) at (3,0) {};
    \draw (v1) -- (v2) -- (v3);
  \end{tikzpicture}\quad\quad\quad\quad
  $\cG_5$: \begin{tikzpicture}[node/.style={circle, draw, inner sep=2pt}]
    \node[node,label=below:$1$] (v1) at (0,0) {};
    \node[node,label=below:$2$] (v2) at (1,0) {};
    \node[node,label=below:$3$] (v3) at (2,0) {};
    \node[node,label=below:$4$] (v4) at (3,0) {};
    \draw (v1) -- (v2) -- (v3);
    \draw[bend left=45] (v1) to (v3);
  \end{tikzpicture}\quad\quad\quad\quad
  $\cG_6$: \begin{tikzpicture}[node/.style={circle, draw, inner sep=2pt}]
    \node[node,label=below:$1$] (v1) at (0,0) {};
    \node[node,label=below:$2$] (v2) at (1,0) {};
    \node[node,label=below:$3$] (v3) at (2,0) {};
    \node[node,label=below:$4$] (v4) at (3,0) {};
    \draw (v1) -- (v2) -- (v3) -- (v4);
  \end{tikzpicture}\\
  $\cG_7$: \begin{tikzpicture}[node/.style={circle, draw, inner sep=2pt}]
    \node[node,label=below:$1$] (v1) at (0,0) {};
    \node[node,label=below:$2$] (v2) at (1,0) {};
    \node[node,label=below:$3$] (v3) at (2,0) {};
    \node[node,label=below:$4$] (v4) at (3,0) {};
    \draw (v1) -- (v2) -- (v3) -- (v4);
    \draw[bend left=45] (v1) to (v3);
  \end{tikzpicture}\quad\quad\quad\quad
  $\cG_8$: \begin{tikzpicture}[node/.style={circle, draw, inner sep=2pt}]
    \node[node,label=below:$1$] (v1) at (0,0) {};
    \node[node,label=below:$2$] (v2) at (1,0) {};
    \node[node,label=below:$3$] (v3) at (2,0) {};
    \node[node,label=below:$4$] (v4) at (3,0) {};
    \draw (v1) -- (v2) -- (v3) -- (v4);
    \draw[bend left=45] (v1) to (v3);
    \draw[bend left=45] (v2) to (v4);
  \end{tikzpicture}\quad\quad\quad\quad
  $\cG_9$: \begin{tikzpicture}[node/.style={circle, draw, inner sep=2pt}]
    \node[node,label=below:$1$] (v1) at (0,0) {};
    \node[node,label=below:$2$] (v2) at (1,0) {};
    \node[node,label=below:$3$] (v3) at (2,0) {};
    \node[node,label=below:$4$] (v4) at (3,0) {};
    \draw (v1) -- (v2) -- (v3) -- (v4);
    \draw[bend left=45] (v1) to (v3);
    \draw[bend left=45] (v2) to (v4);
    \draw[bend left=45] (v1) to (v4);
  \end{tikzpicture}\\
  \caption{Unit interval graphs with four vertices.}
  \label{fig_unit_interval_graphs_4}
  \end{center}
\end{figure}

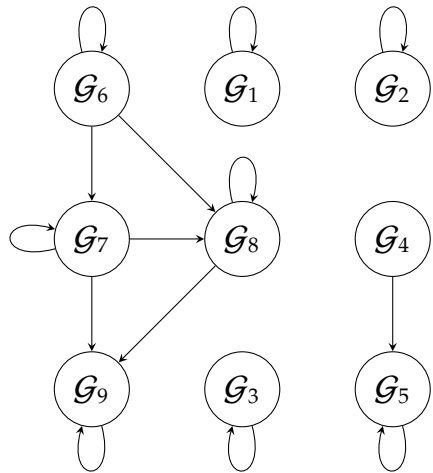
\begin{figure}[htp]
  \begin{center}
  \begin{tikzpicture}[>=stealth, node distance=2cm,
    vertex/.style={circle,draw,minimum size=6pt}]

    \node[vertex] (A6) at (0,4) {$\cG_6$};
    \node[vertex] (A7) at (0,2) {$\cG_7$};
    \node[vertex] (A8) at (2,2) {$\cG_8$};
    \node[vertex] (A9) at (0,0) {$\cG_9$};
    \node[vertex] (A1) at (2,4) {$\cG_1$};
    \node[vertex] (A2) at (4,4) {$\cG_2$};
    \node[vertex] (A3) at (2,0) {$\cG_3$};
    \node[vertex] (A4) at (4,2) {$\cG_4$};
    \node[vertex] (A5) at (4,0) {$\cG_5$};

    \draw[->] (A6) --  (A7);
    \draw[->] (A6) --  (A8);
    \draw[->] (A7) --  (A9);
    \draw[->] (A8) --  (A9);
    \draw[->] (A4) --  (A5);

    \draw[->] (A7) to[loop left]  ();
    \draw[->] (A6) to[loop above]  ();
    \draw[->] (A1) to[loop above]  ();
    \draw[->] (A2) to[loop above]  ();
    \draw[->] (A9) to[loop below] ();
    \draw[->] (A3) to[loop below] ();
    \draw[->] (A5) to[loop below] ();
  \end{tikzpicture}
  \caption{Transition graph for influence graphs with four vertices, see Figure~\ref{fig_unit_interval_graphs_4} for the labeling.}
  \label{fig_transition_4}
  \end{center}
 \end{figure}

 In Example~\ref{example_4} we have distinguished more cases than we have unit interval graphs with three vertices, even if we remove symmetric cases. Also in the context of the freezing time we need a little more detailed information than provided by an influence graph, so that we introduce the following refinement.

 \begin{definition}\label{def_weighted_influence_graph}
   Let $\cP=(X(0),\varepsilon)$ be a BC-process for $n$ agents and $t\in\N$ an arbitrary time step. A \emph{cluster} at time $t$ is a group of agents $\cC\subseteq\{1,\dots,n\}$ with $x_i(t)=x_j(t)$ for all $i,j\in\cC$ and $x_i(t)\neq x_h(t)$ for all $i\in\cC$ and $h\notin\cC$. A partition of the set of agents into clusters $\cC_1(t),\dots,\cC_l(t)$ with $i<j$ for all $i\in\cC_h(t)$ and all $j\in\cC_{h'}(t)$ with $h<h'$ is called the \emph{ordered cluster decomposition}. With this, the \emph{weighted influence graph} $\widetilde{\cG}(t)$ at time $t$ is a graph with vertex set $\cC_1(t),\dots, \cC_l(t)$ where vertex $\cC_i(t)$ has weight $\#\cC_i(t)\in\N$. Each pair $\{\cC_i(t),\cC_j(t)\}$ with $1\le i<j\le l$ is an edge iff $\{a,b\}$ is an edge in the ordinary influence graph $\cG(t)$ where $a\in\cC_i(t)$ and $b\in\cC_j(t)$ are arbitrary.
 \end{definition}

 A cluster of size $n$ occurs iff we have consensus. The sum of the weights of the clusters in an ordered cluster decomposition equals $n$. Each influence graph can correspond to several weighted influence graphs while each weighted influence graph corresponds to a unique influence graph. The possible weighted influence graphs for three agents are depicted in Figure~\ref{fig_weighted_influence_graphs_3}. Instead of labels we write the weight below each vertex (where we may drop weights of one later on). The corresponding transition graph is given in Figure~\ref{fig_transition_weighted_3}. A BC-process has reached its final stable state at time $t\in\N$ if the corresponding weighted influence graph $\widetilde{\cG}(t)$ does not contain any edges. In Figure~\ref{fig_weighted_influence_graphs_3} and Figure~\ref{fig_transition_weighted_3} this is the case for graphs $\widetilde{\cG}_1$, $\widetilde{\cG}_2$, and $\widetilde{\cG}_3$. We observe that for $n=3$ agents there are no loops at unstable states, which is different for $n\ge 4$. From Figure~\ref{fig_transition_weighted_3} we can directly read off the maximum freezing time $F(3)=2$.

 \begin{figure}[htp]
   \begin{center}
    $\widetilde{G}_1$: \begin{tikzpicture}[node/.style={circle, draw, inner sep=2pt}]
    \node[node,label=below:$3$] (v1) at (0,0) {};
  \end{tikzpicture}\quad\quad
  $\widetilde{G}_2$: \begin{tikzpicture}[node/.style={circle, draw, inner sep=2pt}]
    \node[node,label=below:$2$] (v1) at (0,0) {};
    \node[node,label=below:$1$] (v2) at (1,0) {};
  \end{tikzpicture}\quad\quad
  $\widetilde{G}_3$: \begin{tikzpicture}[node/.style={circle, draw, inner sep=2pt}]
    \node[node,label=below:$1$] (v1) at (0,0) {};
    \node[node,label=below:$1$] (v2) at (1,0) {};
    \node[node,label=below:$1$] (v3) at (2,0) {};
  \end{tikzpicture}\quad\quad
  $\widetilde{G}_4$: \begin{tikzpicture}[node/.style={circle, draw, inner sep=2pt}]
    \node[node,label=below:$1$] (v1) at (0,0) {};
    \node[node,label=below:$1$] (v2) at (1,0) {};
    \node[node,label=below:$1$] (v3) at (2,0) {};
    \draw (v1) -- (v2);
  \end{tikzpicture}\\
  $\widetilde{G}_5$: \begin{tikzpicture}[node/.style={circle, draw, inner sep=2pt}]
    \node[node,label=below:$1$] (v1) at (0,0) {};
    \node[node,label=below:$1$] (v2) at (1,0) {};
    \node[node,label=below:$1$] (v3) at (2,0) {};
    \draw (v1) -- (v2) -- (v3);
  \end{tikzpicture}\quad\quad
  $\widetilde{G}_6$: \begin{tikzpicture}[node/.style={circle, draw, inner sep=2pt}]
    \node[node,label=below:$1$] (v1) at (0,0) {};
    \node[node,label=below:$1$] (v2) at (1,0) {};
    \node[node,label=below:$1$] (v3) at (2,0) {};
    \draw (v1) -- (v2) -- (v3);
    \draw[bend left=45] (v1) to (v3);
  \end{tikzpicture}\quad\quad
  $\widetilde{G}_7$: \begin{tikzpicture}[node/.style={circle, draw, inner sep=2pt}]
    \node[node,label=below:$2$] (v1) at (0,0) {};
    \node[node,label=below:$1$] (v2) at (1,0) {};
    \draw (v1) -- (v2);
  \end{tikzpicture}
     \caption{Weighted influence graphs for three agents.}
     \label{fig_weighted_influence_graphs_3}
   \end{center}
 \end{figure}
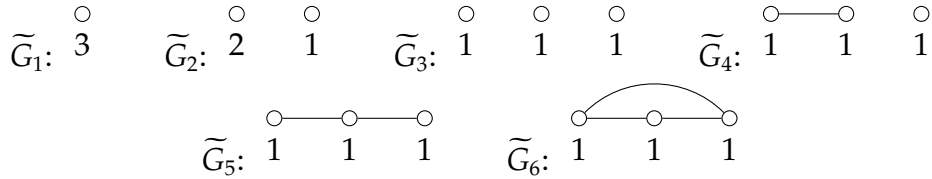

 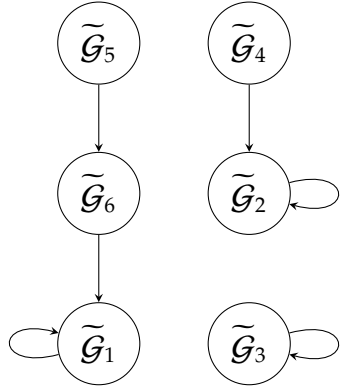
\begin{figure}[htp]
   \begin{center}
     \begin{tikzpicture}[>=stealth, node distance=2cm,
    vertex/.style={circle,draw,minimum size=6pt}]

    \node[vertex] (A1) at (0,0) {$\widetilde{\cG}_1$};
    \node[vertex] (A2) at (2,2) {$\widetilde{\cG}_2$};
    \node[vertex] (A3) at (2,0) {$\widetilde{\cG}_3$};
    \node[vertex] (A4) at (2,4) {$\widetilde{\cG}_4$};
    \node[vertex] (A5) at (0,4) {$\widetilde{\cG}_5$};
    \node[vertex] (A6) at (0,2) {$\widetilde{\cG}_6$};
    \node[vertex] (A7) at (-2,2) {$\widetilde{\cG}_7$};

    \draw[->] (A5) --  (A6);
    \draw[->] (A6) --  (A1);
    \draw[->] (A4) --  (A2);
    \draw[->] (A7) --  (A1);

    \draw[->] (A1) to[loop left]  ();
    \draw[->] (A2) to[loop right]  ();
    \draw[->] (A3) to[loop right]  ();
  \end{tikzpicture}
  \caption{Transition graph for weighted influence graphs with three agents, see Figure~\ref{fig_weighted_influence_graphs_3} for the labeling.}
     \label{fig_transition_weighted_3}
   \end{center}
 \end{figure}

 \begin{remark}
   The techniques described in this section can also be applied to higher dimensional BC-processes in $\R^l$ if we choose a metric $d$ induced by $\Vert\cdot\Vert_1$ or $\Vert\cdot\Vert_\infty$ and additional binary variables to resolve disjunctions and the missing ordering trick for $l\ge 2$. For e.g.\ the Euclidean metric, induced by $ \Vert\cdot\Vert_2$, we would need to replace linear constraints by, possibly non-convex, quadratic constraints. In any case the question for the class of graphs that can be realized as a BC-process is interesting.
 \end{remark}

\section{Maximum freezing time}\label{sec_freezing_time}

As mentioned in Section~\ref{sec_preliminaries} the maximum freezing time $F(n)$ of a BC-process with $n$ agents is upper bounded by $4n\left(n^2+1\right)$. Due to the decomposition in Lemma~\ref{lemma_decompose} 
and $F(n)\ge F(n-1)$, it suffices to consider connected BC-processes in order to determine $F(n)$. By definition we have $F(1)=0$. The discussion of the possible cases in Example~\ref{example_4} directly implies $F(2)=1$ and $F(3)=2$. A configuration showing $F(4)\ge 5$ is given in Example~\ref{example_1}. The latter example is an instance of a so-called \emph{equidistant configuration}, i.e.\ the starting configuration is given by $x_i(0)=\alpha i+\beta$ for some $\alpha,\beta\in \R$. In Example~\ref{example_1} $\alpha$ equals $\varepsilon$. Indeed, $F(n)$ is attained for these configurations for all $n\le 4$, while it is not attained for $5\le n\le 14$ with $n\neq 7$.\footnote{For $\alpha=\varepsilon$ a freezing time of $\tfrac{5n}{6}+O(1)$ was shown in \cite{hegarty2016hegselmann}, and for $n=7$ the equidistant configuration with $\alpha=\tfrac{11}{12}\varepsilon$ attains $F(7)=12$.} The exact values $F(4)=5$,\footnote{See also \cite[Theorem 3]{krause2000discrete}.} $F(5)=7$, $F(6)=9$, and the lower bound $F(7)\ge 12$ are already stated in \cite{kurz2014long}.

Next we want to parameterize all BC-processes for four agents with maximum freezing time. As an initial parameterization we utilize $x_1(0)=0$, $x_2(0)=x$, $x_3(0)=x+y$, $x_4(0)=x+y+z$, and $\varepsilon=1$, where $x,y,z\ge 0$. Due to Example~\ref{example_1} we are only interested in configurations with a freezing time of at least $5$. Since $F(n)\le n-1$ for $n\le 3$, the transition graph in Figure~\ref{fig_transition_4} implies that the only relevant influence graphs are contained in $\left\{\cG_6,\cG_7,\cG_8,\cG_9\right\}$.\footnote{Note that the freezing time is not an invariant of graph equivalence since it does not distinguish between $x_i(t)=x_{i+1}(t)$ and $x_i(t)<x_{i+1}(t)$. To this end we need to use weighted influence graphs.  For the ease of the exposition we ignore this technical subtlety.} First we check how many loops can occur in one of the states $\cG_i$. For $\cG_9$ we observe that after at least one iteration we end up in a consensus. After two times $\cG_7$ we obtain $x_1(2)=\tfrac{25x+14y+3z}{36}$ and $x_4(2)=\tfrac{7x+6y+3z}{8}$, so that
$$
  d\!\left(x_1(2),x_4(2)\right)=\frac{13x+26y+21z}{72}
  <1=\varepsilon
$$
and $x_1(3)=x_2(3)=x_3(3)=x_4(3)$. 
After two times $\cG_6$ we obtain $x_1(2)=\tfrac{7x+2y}{12}$, $x_2(2)=\tfrac{13x+6y+2z}{18}$, $x_3(2)=\tfrac{16x+12y+5z}{18}$, and $x_4(2)=\tfrac{12x+10y+5z}{12}$.

Assume $d\!\left(x_1(2),x_3(2)\right)>1=\varepsilon$ and $d\!\left(x_2(2),x_4(2)\right)>1=\varepsilon$ for a
moment, so that $\cG(2)=\cG_6$. Checking $d\!\left(x_1(3),x_3(3)\right)\le 1=\varepsilon$ and $d\!\left(x_2(3),x_4(3)\right)\le 1=\varepsilon$ we conclude $\cG(3)\in\left\{\cG_8,\cG_9\right\}$. Since we are only interested in a freezing time of at least $5$, we assume $\cG(3)=\cG_8$.\footnote{Alternatively, we may exclude the possibility $\cG(3)=\cG_9$ as follows. $\cG(2) = \cG_6$ requires $11x+18y+10z > 36$ and $10x+18y+11z > 36$. Adding them gives $7x+12y+7z > 24$, i.e.\ $x_4(3)-x_1(3) > 1$, so that $\{1,4\}$ is not an edge.}
Checking $d\!\left(x_1(4),x_4(4)\right)\le 1=\varepsilon$ we conclude $\cG(4)=\cG_9$. The chain $\left(\cG_6,\cG_6,\cG_6,\cG_8,\cG_9\right)$ is indeed attained in Example~\ref{example_1}.

In the following we can assume that $\mathcal G_6$ occurs at most twice. A small computation excludes
the chain $(\mathcal G_6,\mathcal G_6,\mathcal G_7,\mathcal G_7,\mathcal G_9)$, so that we can assume that
$\mathcal G_7$ occurs at most once. Since $\mathcal G_8$ has no loop and there is no transition from
$\mathcal G_7$ to $\mathcal G_8$ in Figure~\ref{fig_transition_4}, none of the remaining chains
reaches freezing time five, so that $(\mathcal G_6,\mathcal G_6,\mathcal G_6,\mathcal G_8,\mathcal G_9)$
is the unique possible chain for freezing time five.\footnote{For freezing time four the possible chains are given by $\left(\cG_6,\cG_6,\cG_8,\cG_9\right)$ and $\left(\cG_6,\cG_6,\cG_7,\cG_9\right)$.}

\begin{example}
  Let $\cP=(X(0),1)$ be a BC-process for five agents with $X(0)=\big(0,1,2,3,3+x\big)$, where $\tfrac{19}{28}< x\le \tfrac{38}{47}$. We recursively compute
  $$
    X(1)=\left(\frac{1}{2},1,2,\frac{8+x}{3},\frac{6+x}{2}\right),
  $$
  $$
    X(2)=\left(\frac{3}{4},\frac{7}{6},\frac{17+x}{9},\frac{46+5x}{18},\frac{34+5x}{12}\right),\text{ and}
  $$
  $$
    X(3)=\left(\frac{23}{24},\frac{137+4x}{108},\frac{101+7x}{54},\frac{262+29x}{108},\frac{194+25x}{72}\right),
  $$
  where $\cG(0)=\cG(1)=\cG(2)$ is a path $1-2-3-4-5$. For $\cG(3)=\cG(4)$ there is an additional edge between agents $3$ and $5$, so that
  $$
    X(4)=\left(\frac{481+8x}{432},\frac{295+12x}{216},\frac{1784+169x}{864},\frac{1510+161x}{648},\frac{1510+161x}{648}\right)\text{ and}
  $$
  $$
    X(5)=\left(\frac{1071+32x}{864},\frac{3926+233x}{2592},\frac{20972+1939x}{10368},\frac{17432+1795x}{7776},\frac{17432+1795x}{7776}\right).
  $$
  In $\cG(5)$ the only missing edges are $\{1,4\}$ and $\{1,5\}$.
  With this we compute
  $$
    x_1(6)=\frac{49528+3255x}{31104}\quad\text{and}\quad
    x_5(6)=\frac{249484 + 22973x}{124416}.
  $$
  Since $0<x_5(6)-x_1(6)\le 1$ we read a consensus at time $t=7$.
\end{example}
Up to mirror symmetry there is a unique sequence of influence graphs for five agents that can be realized by a BC-process with maximum freezing time $F(5)=7$.

\begin{example}
  The following sequence of influence graphs $\cG(t)$ can be realized by a BC-process, which shows $F(6)\ge 9$.

  \medskip

  \noindent
  $\cG(0)$: 
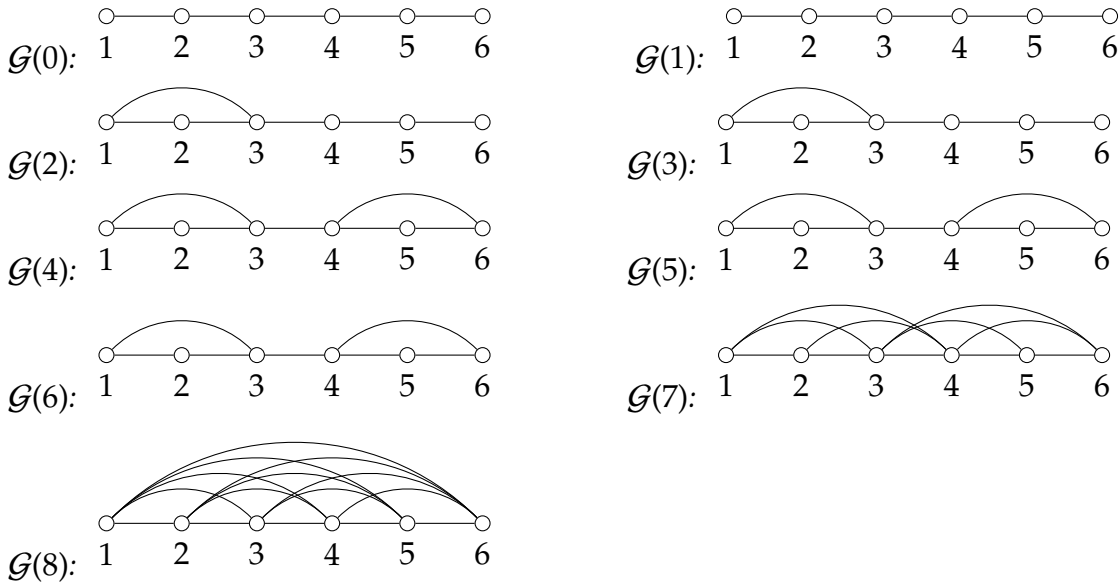
\begin{tikzpicture}[node/.style={circle, draw, inner sep=2pt}]
    \node[node,label=below:$1$] (v1) at (0,0) {};
    \node[node,label=below:$2$] (v2) at (1,0) {};
    \node[node,label=below:$3$] (v3) at (2,0) {};
    \node[node,label=below:$4$] (v4) at (3,0) {};
    \node[node,label=below:$5$] (v5) at (4,0) {};
    \node[node,label=below:$6$] (v6) at (5,0) {};
    \draw (v1) -- (v2) -- (v3) -- (v4) -- (v5) -- (v6);
  \end{tikzpicture} \quad\quad\quad\quad
  $\cG(1)$:  \begin{tikzpicture}[node/.style={circle, draw, inner sep=2pt}]
    \node[node,label=below:$1$] (v1) at (0,0) {};
    \node[node,label=below:$2$] (v2) at (1,0) {};
    \node[node,label=below:$3$] (v3) at (2,0) {};
    \node[node,label=below:$4$] (v4) at (3,0) {};
    \node[node,label=below:$5$] (v5) at (4,0) {};
    \node[node,label=below:$6$] (v6) at (5,0) {};
    \draw (v1) -- (v2) -- (v3) -- (v4) -- (v5) -- (v6);
  \end{tikzpicture}\\
  $\cG(2)$:  \begin{tikzpicture}[node/.style={circle, draw, inner sep=2pt}]
    \node[node,label=below:$1$] (v1) at (0,0) {};
    \node[node,label=below:$2$] (v2) at (1,0) {};
    \node[node,label=below:$3$] (v3) at (2,0) {};
    \node[node,label=below:$4$] (v4) at (3,0) {};
    \node[node,label=below:$5$] (v5) at (4,0) {};
    \node[node,label=below:$6$] (v6) at (5,0) {};
    \draw (v1) -- (v2) -- (v3) -- (v4) -- (v5) -- (v6);
    \draw[bend left=45] (v1) to (v3);
  \end{tikzpicture}\quad\quad\quad\quad
  $\cG(3)$:  \begin{tikzpicture}[node/.style={circle, draw, inner sep=2pt}]
    \node[node,label=below:$1$] (v1) at (0,0) {};
    \node[node,label=below:$2$] (v2) at (1,0) {};
    \node[node,label=below:$3$] (v3) at (2,0) {};
    \node[node,label=below:$4$] (v4) at (3,0) {};
    \node[node,label=below:$5$] (v5) at (4,0) {};
    \node[node,label=below:$6$] (v6) at (5,0) {};
    \draw (v1) -- (v2) -- (v3) -- (v4) -- (v5) -- (v6);
    \draw[bend left=45] (v1) to (v3);
  \end{tikzpicture}\\
  $\cG(4)$:  \begin{tikzpicture}[node/.style={circle, draw, inner sep=2pt}]
    \node[node,label=below:$1$] (v1) at (0,0) {};
    \node[node,label=below:$2$] (v2) at (1,0) {};
    \node[node,label=below:$3$] (v3) at (2,0) {};
    \node[node,label=below:$4$] (v4) at (3,0) {};
    \node[node,label=below:$5$] (v5) at (4,0) {};
    \node[node,label=below:$6$] (v6) at (5,0) {};
    \draw (v1) -- (v2) -- (v3) -- (v4) -- (v5) -- (v6);
    \draw[bend left=45] (v1) to (v3);
    \draw[bend left=45] (v4) to (v6);
  \end{tikzpicture}\quad\quad\quad\quad
  $\cG(5)$:  \begin{tikzpicture}[node/.style={circle, draw, inner sep=2pt}]
    \node[node,label=below:$1$] (v1) at (0,0) {};
    \node[node,label=below:$2$] (v2) at (1,0) {};
    \node[node,label=below:$3$] (v3) at (2,0) {};
    \node[node,label=below:$4$] (v4) at (3,0) {};
    \node[node,label=below:$5$] (v5) at (4,0) {};
    \node[node,label=below:$6$] (v6) at (5,0) {};
    \draw (v1) -- (v2) -- (v3) -- (v4) -- (v5) -- (v6);
    \draw[bend left=45] (v1) to (v3);
    \draw[bend left=45] (v4) to (v6);
  \end{tikzpicture}\\
  $\cG(6)$:  \begin{tikzpicture}[node/.style={circle, draw, inner sep=2pt}]
    \node[node,label=below:$1$] (v1) at (0,0) {};
    \node[node,label=below:$2$] (v2) at (1,0) {};
    \node[node,label=below:$3$] (v3) at (2,0) {};
    \node[node,label=below:$4$] (v4) at (3,0) {};
    \node[node,label=below:$5$] (v5) at (4,0) {};
    \node[node,label=below:$6$] (v6) at (5,0) {};
    \draw (v1) -- (v2) -- (v3) -- (v4) -- (v5) -- (v6);
    \draw[bend left=45] (v1) to (v3);
    \draw[bend left=45] (v4) to (v6);
  \end{tikzpicture}\quad\quad\quad\quad
  $\cG(7)$:  \begin{tikzpicture}[node/.style={circle, draw, inner sep=2pt}]
    \node[node,label=below:$1$] (v1) at (0,0) {};
    \node[node,label=below:$2$] (v2) at (1,0) {};
    \node[node,label=below:$3$] (v3) at (2,0) {};
    \node[node,label=below:$4$] (v4) at (3,0) {};
    \node[node,label=below:$5$] (v5) at (4,0) {};
    \node[node,label=below:$6$] (v6) at (5,0) {};
    \draw (v1) -- (v2) -- (v3) -- (v4) -- (v5) -- (v6);
    \draw[bend left=45] (v1) to (v3);
    \draw[bend left=45] (v1) to (v4);
    \draw[bend left=45] (v2) to (v4);
    \draw[bend left=45] (v4) to (v6);
    \draw[bend left=45] (v3) to (v5);
    \draw[bend left=45] (v3) to (v6);
  \end{tikzpicture}\\
  $\cG(8)$:  \begin{tikzpicture}[node/.style={circle, draw, inner sep=2pt}]
    \node[node,label=below:$1$] (v1) at (0,0) {};
    \node[node,label=below:$2$] (v2) at (1,0) {};
    \node[node,label=below:$3$] (v3) at (2,0) {};
    \node[node,label=below:$4$] (v4) at (3,0) {};
    \node[node,label=below:$5$] (v5) at (4,0) {};
    \node[node,label=below:$6$] (v6) at (5,0) {};
    \draw (v1) -- (v2) -- (v3) -- (v4) -- (v5) -- (v6);
    \draw[bend left=45] (v1) to (v3);
    \draw[bend left=45] (v1) to (v4);
    \draw[bend left=45] (v1) to (v5);
    \draw[bend left=45] (v1) to (v6);
    \draw[bend left=45] (v2) to (v4);
    \draw[bend left=45] (v2) to (v5);
    \draw[bend left=45] (v2) to (v6);
    \draw[bend left=45] (v4) to (v6);
    \draw[bend left=45] (v3) to (v5);
    \draw[bend left=45] (v3) to (v6);
  \end{tikzpicture}
\end{example}

Up to mirror symmetry there are exactly two sequences of influence graphs for six agents that can be realized by a BC-process with maximum freezing time $F(6)=9$.

We can algorithmically determine $F(n)$ using the transition graph for $n$ agents. First we determine for each state $\cG'$ the maximum number of iterations a BC-process can remain in the same state.\footnote{We can easily recognize stable states by checking whether all connected components are complete graphs, i.e.\ graphs with all possible edges.} The longest path in the transition graph with bounded loops is an upper bound for $F(n)$. In order to verify $F(n)<k$ we only need to check that no path consisting of $k-1$ edges is realizable as a BC-process. We may also tabulate some non-realizable subpaths to speed up the search. We did this for $7$ agents and obtained $F(7)=12$. There are exactly two corresponding realizable sequences of influence graphs given by
$\cG_1^4 \cG_2^4 \cG_3^3 \cG_4^1$ and
$\cG_1^4 \cG_2^5 \cG_3^2 \cG_4^1$, where the exponent denotes repetition, $\cG_4$ is the complete graph on seven vertices, and $\cG_1,\cG_2,\cG_3$ are given as follows.

\bigskip

\noindent
$\cG_1$: 
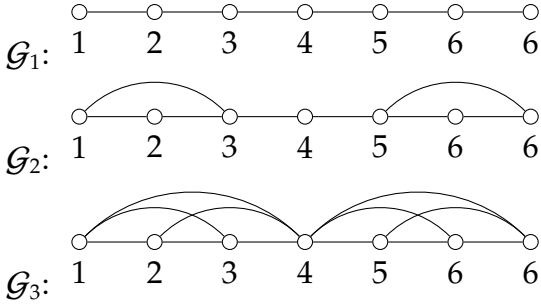
\begin{tikzpicture}[node/.style={circle, draw, inner sep=2pt}]
    \node[node,label=below:$1$] (v1) at (0,0) {};
    \node[node,label=below:$2$] (v2) at (1,0) {};
    \node[node,label=below:$3$] (v3) at (2,0) {};
    \node[node,label=below:$4$] (v4) at (3,0) {};
    \node[node,label=below:$5$] (v5) at (4,0) {};
    \node[node,label=below:$6$] (v6) at (5,0) {};
    \node[node,label=below:$7$] (v7) at (6,0) {};
    \draw (v1) -- (v2) -- (v3) -- (v4) -- (v5) -- (v6)-- (v7);
  \end{tikzpicture}\\
  $\cG_2$: \begin{tikzpicture}[node/.style={circle, draw, inner sep=2pt}]
    \node[node,label=below:$1$] (v1) at (0,0) {};
    \node[node,label=below:$2$] (v2) at (1,0) {};
    \node[node,label=below:$3$] (v3) at (2,0) {};
    \node[node,label=below:$4$] (v4) at (3,0) {};
    \node[node,label=below:$5$] (v5) at (4,0) {};
    \node[node,label=below:$6$] (v6) at (5,0) {};
    \node[node,label=below:$7$] (v7) at (6,0) {};
    \draw (v1) -- (v2) -- (v3) -- (v4) -- (v5) -- (v6)-- (v7);
    \draw[bend left=45] (v1) to (v3);
    \draw[bend left=45] (v5) to (v7);
  \end{tikzpicture}\\
  $\cG_3$: \begin{tikzpicture}[node/.style={circle, draw, inner sep=2pt}]
    \node[node,label=below:$1$] (v1) at (0,0) {};
    \node[node,label=below:$2$] (v2) at (1,0) {};
    \node[node,label=below:$3$] (v3) at (2,0) {};
    \node[node,label=below:$4$] (v4) at (3,0) {};
    \node[node,label=below:$5$] (v5) at (4,0) {};
    \node[node,label=below:$6$] (v6) at (5,0) {};
    \node[node,label=below:$7$] (v7) at (6,0) {};
    \draw (v1) -- (v2) -- (v3) -- (v4) -- (v5) -- (v6)-- (v7);
    \draw[bend left=45] (v1) to (v3);
    \draw[bend left=45] (v1) to (v4);
    \draw[bend left=45] (v2) to (v4);
    \draw[bend left=45] (v4) to (v7);
    \draw[bend left=45] (v4) to (v6);
    \draw[bend left=45] (v5) to (v7);
  \end{tikzpicture}

\section{Maximum fragmentation}\label{sec_fragmentation}

Given a BC-process $\cP=(X(0),\varepsilon)$ with $n$ agents we denote the number of connected components of $\cG(F(\cP))$ by $S(\cP)$, i.e., the number of connected components of the influence graph after the configuration has frozen. By $S(n)$ we denote the maximum value of $S(\cP)$ over all connected BC-processes with $n$ agents, i.e.\ those BC-processes where $\cG(0)$ is connected.\footnote{If the $x_i(0)$ are pairwise different for $1\le i\le n$ and $\varepsilon$ is sufficiently small, then there are no edges in $\cG(0)$ and we have $n$ connected components. Also in the more general situation where not all connected components consist of a single vertex we can use Lemma~\ref{lemma_decompose} to decompose the configuration into several connected ones.} So, $S(\cP)$ measures the final degree of fragmentation, i.e.\ how much the stable state deviates from a consensus where $S(\cP)=1$. Clearly, we have $S(n)\le n$, which is indeed attained for $n=1$. The discussion in Example~\ref{example_4} shows $S(1)=S(2)=S(3)=1$. From the transition graph in Figure~\ref{fig_transition_4} we also conclude $S(4)=1$, cf.~\cite[Theorem 3(i)]{krause2000discrete}. We remark that consensus, i.e.\ $S(\cP)=1$, can be characterized as follows: The process $\cP$ ends in a consensus iff $\cG(t)$ is connected, i.e.\ $X(t)$ is an $\varepsilon$-profile, for all $t\in\N$. Necessity was shown in \cite[Theorem 2]{krause2000discrete} and sufficiency in \cite[Theorem 8]{dittmer2001consensus}. In our setting both directions follow directly from Lemma~\ref{lemma_decompose}: a disconnected $\cG(t)$ stays disconnected for all later times, so that $S(\cP)\ge 2$, while a connected component of a frozen influence graph is a complete graph and all its agents share the same opinion.

\begin{example}(Cf.~\cite[Theorem 3(ii)]{krause2000discrete})
   \label{example_8}
   Consider the BC-process $\cP=(X(0),1)$ with a starting configuration given by $x_1(0)=x_2(0)=0$, $x_3(0)=1$, $x_4(0)=2$, and $x_5(0)=3$. Note that $\cG(0)$ is connected. After one iteration we obtain $x_1(1)=x_2(1)=\tfrac{1}{3}$, $x_3(1)=\tfrac{3}{4}$, $x_4(1)=2$, and $x_5(1)=\tfrac{5}{2}$, so that $\cP$ freezes after two iterations with $x_1(2)=x_2(2)=x_3(2)=\tfrac{17}{36}$ and $x_4(2)=x_5(2)=\tfrac{9}{4}$. Thus, we have $S(5)\ge 2$.

   Similarly, consider the BC-process $\cP=(X(0),1)$ with a starting configuration given by $x_1(0)=x_2(0)=0$, $x_3(0)=1$, $x_4(0)=2$, and $x_5(0)=x_6(0)=3$. Again, $\cG(0)$ is connected. After one iteration we obtain $X(1)=\left(\tfrac{1}{3},\tfrac{1}{3},\tfrac{3}{4},\tfrac{9}{4},\tfrac{8}{3},\tfrac{8}{3}\right)$. The BC-process freezes at
   $X(2)=\left(\tfrac{17}{36},\tfrac{17}{36},\tfrac{17}{36},\tfrac{91}{36},\tfrac{91}{36},\tfrac{91}{36}\right)$ with two clusters of size three. Thus, we have $S(6)\ge 2$.
\end{example}

By examining the transition graphs for five and for six agents we can verify $S(5)=S(6)=2$. Note that we are not aware of any proof showing $S(n+1)\ge S(n)$ for all $n\in\N$, while we also do not know any counter-example to this monotonicity property.

\begin{lemma}
  \label{lemma_seg_bound_1}
  Let $\cP=(X(0),1)$ be a BC-process, $s=\#I_1(0)-1$,
  and $t=\#I_2(0)-1$. Assume $x_1(0)\le\dots \le x_n(0)$. If $s\ge 1$ and $x_2(1)-x_1(1)>1$, then we have
  \begin{equation}
    t> 2(s+1)+\tfrac{1}{s}\ge 5,
  \end{equation}
  so that $t\ge 6$.
\end{lemma}
\begin{proof}
  Using the parameterization $x_i(0)=\sum_{j=1}^{i-1} \delta_j$ with non-negative $\delta_j$, assuming $x_1(0)=0$, we compute
  \begin{equation}
    x_1(1)= \frac{\sum_{i=1}^s (s+1-i)\delta_i}{s+1}
  \end{equation}
  and
  \begin{equation}
    x_2(1)= \frac{\sum_{i=1}^t (t+1-i)\delta_i}{t+1} ,
  \end{equation}
  so that
  \begin{eqnarray*}
    x_2(1)-x_1(1)&=&\frac{t-s}{(s+1)(t+1)}\sum_{i=1}^s i\delta_i
    \,+\, \frac{\sum\limits_{i=s+1}^t (t+1-i)\delta_i}{t+1}\\
    &\le& \frac{(t-s)s}{(s+1)(t+1)}\cdot \underset{x_{s+1}(0)-x_1(0)}{\underbrace{\sum_{i=1}^s \delta_i}} \,+\,\frac{t-s}{t+1} \cdot \underset{x_{t+1}(0)-x_{s+1}(0)}{\underbrace{\sum_{i=s+1}^t \delta_i}} \\
    &\le& \frac{(t-s)s}{(s+1)(t+1)}+ \frac{t-s}{t+1}=
    \frac{(t-s)(2s+1)}{(s+1)(t+1)}.
  \end{eqnarray*}
  Since $x_2(1)-x_1(1)>1$ and $s\in\N_{>0}$ the statement follows.
\end{proof}

\begin{lemma}
  \label{lemma_seg_bound_2}
  Let $\cP=(X(0),1)$ be a BC-process with $n\ge 3$ agents satisfying
  $x_1(0)\le\dots\le x_n(0)$. Let $2\le i\le n-1$ and assume that
  $i-1,\,i+1\in I_i(0)$. Let
  $s_1,s_2,a_1,a_2,b_1,b_2\in\mathbb{N}$ be such that
  \begin{align*}
  I_{i-1}(0)&=\{i-a_1,\dots,i+a_2\},\\
  I_{i}(0)&=\{i-s_1,\dots,i+s_2\},\\
  I_{i+1}(0)&=\{i-b_1,\dots,i+b_2\}.
  \end{align*}
  If $d_1:=x_i(1)-x_{i-1}(1)>1$ and $d_2:=x_{i+1}(1)-x_i(1)>1$, then
  $a_1+b_2\ge 6$.
\end{lemma}
\begin{proof}
Write $p_j:=x_j(0)-x_i(0)$ for $1\le j\le n$, so that $p_i=0$ and
$p_1\le\dots\le p_n$. For a non-empty set $S\subseteq A$ of agents let
\[
  \mu(S):=\frac{1}{\#S}\sum_{j\in S}p_j
\]
be the arithmetic mean of the \emph{shifted} opinions $p_j$ over $S$. By the updating rule and since shifting all opinions by $x_i(0)$ shifts every mean by $x_i(0)$,
we have $x_j(1)-x_i(0)=\mu\bigl(I_j(0)\bigr)$ for all $j\in A$. Hence
\begin{equation}\label{eq:d1d2}
  d_1=\mu(I_i(0))-\mu(I_{i-1}(0)),\qquad
  d_2=\mu(I_{i+1}(0))-\mu(I_i(0)),\text{ and }
  d_1+d_2=\mu(I_{i+1}(0))-\mu(I_{i-1}(0)).
\end{equation}
Since $i\pm1\in I_i(0)$ we have $p_{i-1}\in[-1,0]$ and
$p_{i+1}\in[0,1]$, and $s_1,s_2\ge 1$. By Lemma~\ref{lemma_insider_interval} and the
ordering, the left ends and the right ends of the windows
$I_{i-1}(0),I_i(0),I_{i+1}(0)$ are non-decreasing, i.e.\
$a_1\ge s_1\ge b_1\ge 0$ and $0\le a_2\le s_2\le b_2$; in particular
$a_1,b_2\ge 1$. Moreover, $a_2\ge 1$ iff $i+1\in I_{i-1}(0)$, which is equivalent to $\left|p_{i+1}-p_{i-1}\right|\le 1$. Similarly, this holds iff $i-1\in I_{i+1}(0)$, which is equivalent to $b_1\ge 1$. I.e.\ we have $a_2\ge 1$ iff $b_1\ge 1$.

Assume $a_2\ge1$, so that $b_1\ge 1$.
Put $K:=I_{i-1}(0)\cap I_{i+1}(0)=\{i-b_1,\dots,i+a_2\}$,
$L:=I_{i-1}(0)\setminus I_{i+1}(0)=\{i-a_1,\dots,i-b_1-1\}$ and
$R:=I_{i+1}(0)\setminus I_{i-1}(0)=\{i+a_2+1,\dots,i+b_2\}$. These sets are
pairwise disjoint, $I_{i-1}(0)=L\cup K$, $I_{i+1}(0)=K\cup R$,
$\#L=a_1-b_1$, $\#R=b_2-a_2$ and $\#K=a_2+b_1+1\ge 3$. By~\eqref{eq:d1d2}, we have
$d_1+d_2=\mu(K\cup R)-\mu(K\cup L)$. For $\rho_R:=\#R/(\#K+\#R)$ and
$\rho_L:=\#L/(\#K+\#L)$ we have $\mu(K\cup R)=\mu(K)+\rho_R\bigl(\mu(R)-\mu(K)\bigr)$
and $\mu(K\cup L)=\mu(K)-\rho_L\bigl(\mu(K)-\mu(L)\bigr)$, so
\[
  d_1+d_2=\rho_R\bigl(\mu(R)-\mu(K)\bigr)+\rho_L\bigl(\mu(K)-\mu(L)\bigr),
\]
where a term belonging to an empty set $L$ or $R$ is omitted.
Every index in $L$ is smaller
than every index in $K$, which is smaller than every index in $R$, i.e.\
$L, K, R$ are contiguous index blocks and the non-decreasing $p_j$ imply $\mu(L)\le\mu(K)\le\mu(R)$. Every $j\in K$ is an
insider of both $i-1$ and $i+1$, so $p_{i+1}-1\le p_j\le p_{i-1}+1$ and thus
$p_{i+1}-1\le\mu(K)\le p_{i-1}+1$. Every $j\in R$ is an insider of $i+1$, so
$\mu(R)\le p_{i+1}+1$; every $j\in L$ is an insider of $i-1$, so
$\mu(L)\ge p_{i-1}-1$. Therefore
\[
  0\le\mu(R)-\mu(K)\le(p_{i+1}+1)-(p_{i+1}-1)=2,\qquad
  0\le\mu(K)-\mu(L)\le(p_{i-1}+1)-(p_{i-1}-1)=2 .
\]
With $f(z):=z/(3+z)$, which is increasing
and concave on $[0,\infty)$, $\#K\ge3$ gives $\rho_R\le f(\#R)$ and
$\rho_L\le f(\#L)$. Thus, we have
\[
  d_1+d_2\le 2\bigl(f(b_2-a_2)+f(a_1-b_1)\bigr).
\]
Assume $a_1+b_2\le5$. As $a_2,b_1\ge1$ we get
$(b_2-a_2)+(a_1-b_1)\le 3$, so by concavity and monotonicity
\[
  f(b_2-a_2)+f(a_1-b_1)\le 2f\!\left(\tfrac32\right)=\tfrac23,
\]
whence $d_1+d_2\le\tfrac43<2$, contradicting $d_1,d_2>1$. Thus $a_1+b_2\ge6$.

Assume $a_2=b_1=0$ and $a_1+b_2\le 5$. Using $s_1,s_2\ge 1$
the tuple $\tau=\left(s_1,s_2,a_1,0,0,b_2\right)$ is one of the $35$ tuples satisfying
$a_1,b_2\ge 1$, $a_1+b_2\le 5$, $1\le s_1\le a_1$, and $1\le s_2\le b_2$.
Let $W:=\{i-a_1,\dots,i+b_2\}$ and for $c\in\{i-1,i,i+1\}$ let
$\left[l_c,r_c\right]$ be the index window of $I_c(0)$, i.e.\
$\left(l_{i-1},r_{i-1}\right)=\left(i-a_1,i\right)$, $\left(l_i,r_i\right)=\left(i-s_1,i+s_2\right)$, and
$\left(l_{i+1},r_{i+1}\right)=\left(i,i+b_2\right)$.
Let $\mathcal{Q}(\tau)\subset\R^{W}$ be the polyhedron of all
$p=\left(p_j\right)_{j\in W}$ with
\begin{enumerate}
\item[(i)] $p_i=0$ and $p_j\le p_{j+1}$ for all $j,j+1\in W$;
\item[(ii)] $|p_j-p_c|\le 1$ for $c\in\{i-1,i,i+1\}$ and $l_c\le j\le r_c$;
\item[(iii)] $p_c-p_j\ge 1$ for $j\in W$, $j<l_c$, and
$p_j-p_c\ge 1$ for $j\in W$, $j>r_c$, where $c\in\{i-1,i,i+1\}$.
\end{enumerate}
Define the linear forms
\[
d_1(p)=\frac{1}{s_1+s_2+1}\sum_{j=i-s_1}^{i+s_2}p_j-\frac{1}{a_1+1}\sum_{j=i-a_1}^{i}p_j,
\qquad
d_2(p)=\frac{1}{b_2+1}\sum_{j=i}^{i+b_2}p_j-\frac{1}{s_1+s_2+1}\sum_{j=i-s_1}^{i+s_2}p_j .
\]
By the definition of insiders and outsiders (the strict inequalities
$>1$ for outsiders are relaxed to $\ge 1$), we have $\left(p_j\right)_{j\in W}\in \mathcal{Q}(\tau)$, $d_1=d_1(p)$ and $d_2=d_2(p)$. Thus, it suffices to verify
\[
\max\left\{\,\lambda:\ p\in\mathcal{Q}(\tau),\ \lambda\le d_1(p),\ \lambda\le d_2(p)\,\right\}\le 1
\]
for each tuple $\tau$ using an LP solver. The maximum value $1$ is attained exactly for $$\left(s_1,s_2,a_1,a_2,b_1,b_2\right)\in\{(1,1,2,0,0,2),(1,1,2,0,0,3),(1,1,3,0,0,2)\}.$$ Thus, $d_1>1$ and $d_2>1$ cannot hold both -- contradiction, so that $a_1+b_2\ge 6$.
\end{proof}

The assumption $i-1,\,i+1\in I_i(0)$ in Lemma~\ref{lemma_seg_bound_2} e.g.\ holds if $\cG(0)$ is connected. For $\left(s_1,s_2,a_1,a_2,b_1,b_2\right)=(1,1,3,0,0,3)$ the
linear program in the last step has optimal value $\tfrac{5}{4}>1$. The tuple is realizable with $a_1+b_2=6$, see 
Example~\ref{ex_s_9_ge_4} for $i=6$.

The cases of Lemma~\ref{lemma_seg_bound_1} and Lemma~\ref{lemma_seg_bound_2} cannot occur for $n\le 6$ agents. So, if we would end up with three clusters they all have to be of size $2$ and $n=6$. Since $S(1)=S(2)=S(3)=S(4)=1$ we can assume that all clusters are created at the same time. As influence graph only a path of length $5$ might be possible. However, we can also exclude this possibility by a small computation so that we end up with a theoretical proof for $S(5)=S(6)=2$.
Next we state two constructions for a linear number of connected components in the stable state.

\begin{lemma}
  We have $S(5k+2)\ge 2k+1$ for all $k\ge 1$.
\end{lemma}
\begin{proof}
  We consider a BC-process $\cP=(X(0),1)$ with $x_i(0)\in\{0,1,\dots,4k\}$. More precisely, there are exactly two agents with $x_i=4j$ for all $0\le j\le k$ and exactly one agent with $x_i=h$ for all $0\le h\le 4k$ that are not divisible by $4$. Setting $K:=4k$, we compute
  \begin{eqnarray*}
    \tiny
    X(1)&\!\!\!\!=\!\!\!\!&\big(\tfrac{1}{3},\tfrac{1}{3},\tfrac{3}{4},\underset{1}{\underbrace{2,3\!+\!\tfrac{1}{4},4,4,5\!-\!\tfrac{1}{4}}},\underset{2}{\underbrace{6,7\!+\!\tfrac{1}{4},8,8,9\!-\!\tfrac{1}{4}}},\dots,\underset{k-1}{\underbrace{K\!-\!6,K\!-\!5\!+\!\tfrac{1}{4},K\!-\!4,K\!-\!4,K\!-\!3\!-\!\tfrac{1}{4}}},\\
    &&K\!-\!2,K\!-\!\tfrac{3}{4},K\!-\!\tfrac{1}{3},K\!-\!\tfrac{1}{3}\big)\\
    X(2)&\!\!\!\!=\!\!\!\!&\big(\tfrac{17}{36},\tfrac{17}{36},\tfrac{17}{36},\underset{1}{\underbrace{2,3\!+\!\tfrac{3}{4},4,4,5\!-\!\tfrac{3}{4}}},
    \dots,\underset{k-1}{\underbrace{K\!-\!6,K\!-\!5\!+\!\tfrac{3}{4},K\!-\!4,K\!-\!4,K\!-\!3\!-\!\tfrac{3}{4}}},K\!-\!2,K\!-\!\tfrac{17}{36},K\!-\!\tfrac{17}{36},K\!-\!\tfrac{17}{36}\big)\\
    X(3)&\!\!\!\!=\!\!\!\!&\big(
    \tfrac{17}{36},\tfrac{17}{36},\tfrac{17}{36},\underset{1}{\underbrace{2,4,4,4,4}},
    \dots,\underset{k-1}{\underbrace{K\!-\!6,K\!-\!4,K\!-\!4,K\!-\!4,K\!-\!4\!}},K\!-\!2,K\!-\!\tfrac{17}{36},K\!-\!\tfrac{17}{36},K\!-\!\tfrac{17}{36}\big),
  \end{eqnarray*}
  so that the freezing time is at most three and there are $2k+1$ components in $\cG(3)$. The sizes of the components are given by
  $$
    \big(3,\underset{1}{\underbrace{1,4}},\dots,\underset{k-1}{\underbrace{1,4}},1,3\big).
  $$
\end{proof}

\begin{lemma}
  \label{lemma_segregation_fraction_one_half}
  We have $S(4k+9)\ge 2k+3$ for all $k\ge 0$.
\end{lemma}
\begin{proof}
  We consider a BC-process $\cP=(X(0),1)$ with
  $$
    X(0)=\big(0,1,1+\tfrac{1}{2},2,\underset{1}{\underbrace{3,4,4+\tfrac{1}{2},5}},\dots,\underset{k}{\underbrace{K,K+1,K+1+\tfrac{1}{2},K+2}},K+3,K+4,K+4+\tfrac{1}{2},K+5,K+6\big),
  $$
  where $K:=3k$. With this, we compute
  \begin{eqnarray*}
    X(1)&\!\!\!=\!\!\!& \big(
    \tfrac{1}{2},1\!+\!\tfrac{1}{8},1\!+\!\tfrac{1}{2},2\!-\!\tfrac{1}{8},\underset{1}{\underbrace{3,4\!+\!\tfrac{1}{8},4\!+\!\tfrac{1}{2},5\!-\!\tfrac{1}{8}}},\dots,\underset{k}{\underbrace{K,K\!+\!1\!+\!\tfrac{1}{8},K\!+\!1\!+\!\tfrac{1}{2},K\!+\!2\!-\!\tfrac{1}{8}}},\\
    && \!K\!+\!3,K\!+\!4+\tfrac{1}{8},K\!+\!4+\tfrac{1}{2},K\!+\!5\!-\!\tfrac{1}{8},K\!+\!5\!+\!\tfrac{1}{2}
    \big)\\
    X(2)&\!\!\!=\!\!\!&\big(1\!+\!\tfrac{1}{24},1\!+\!\tfrac{1}{4},1\!+\!\tfrac{1}{4},1\!+\!\tfrac{1}{2},
    \underset{1}{\underbrace{3,4\!+\!\tfrac{1}{2},4\!+\!\tfrac{1}{2},4\!+\!\tfrac{1}{2}}},\dots,\underset{k}{\underbrace{K,K\!+\!1\!+\!\tfrac{1}{2},K\!+\!1\!+\!\tfrac{1}{2},K\!+\!1+\!\tfrac{1}{2}}}, \\
    && K\!+\!3,K\!+\!5\!-\!\tfrac{1}{2},K\!+\!5\!-\!\tfrac{1}{4},K\!+\!5\!-\!\tfrac{1}{4},K\!+\!5\!-\!\tfrac{1}{24}\big) \\
    X(3)&\!\!\!=\!\!\!&\big(1\!+\!\tfrac{25}{96},1\!+\!\tfrac{25}{96},1\!+\!\tfrac{25}{96},1\!+\!\tfrac{25}{96},
    \underset{1}{\underbrace{3,4\!+\!\tfrac{1}{2},4\!+\!\tfrac{1}{2},4\!+\!\tfrac{1}{2}}},\dots,\underset{k}{\underbrace{K,K\!+\!1\!+\!\tfrac{1}{2},K\!+\!1\!+\!\tfrac{1}{2},K\!+\!1+\!\tfrac{1}{2}}},\\
    &&K\!+\!3,K\!+\!5\!-\tfrac{25}{96},K\!+\!5\!-\tfrac{25}{96},K\!+\!5\!-\tfrac{25}{96},K\!+\!5\!-\tfrac{25}{96}
    \big),
  \end{eqnarray*}
  so that the freezing time is three and there are $2k+3$ components in $\cG(3)$. The sizes of the components are given by
  $$
    \big(4,\underset{1}{\underbrace{1,3}},\dots,\underset{k}{\underbrace{1,3}},1,4\big).
  $$
\end{proof}
Note that Lemma~\ref{lemma_segregation_fraction_one_half} yields $S(9)\ge 3$ for $k=0$, while Example~\ref{ex_s_9_ge_4} attains $S(9)=4$.
As a fun fact we remark that, in 2025, ChatGPT claimed a proof for $S(n)=\lceil\tfrac{n}{2}\rceil$  based on the wrong assumption that no cluster of size one can occur in the middle of an initially connected starting configuration (and a wrong construction for the lower bound).

\begin{example}
  \label{ex_three_ones}
  Let $\cP=(X(0),1)$ be a BC-process with $13$ agents, where $X(0)=(0,0,0,1,1,2,3,4,5,5,6,6,6)$. We compute
  $X(1)=\big(\tfrac{2}{5},\tfrac{2}{5},\tfrac{2}{5},\tfrac{2}{3},\tfrac{2}{3},\tfrac{7}{4},3,\tfrac{17}{4},\tfrac{16}{3},\tfrac{16}{3},\tfrac{28}{5},\tfrac{28}{5},\tfrac{28}{5}\big)$ and $X(2)=\big(\tfrac{38}{75},\tfrac{38}{75},\tfrac{38}{75},\tfrac{38}{75},\tfrac{38}{75},\tfrac{7}{4},3,\tfrac{17}{4},6-\tfrac{38}{75},6-\tfrac{38}{75},6-\tfrac{38}{75},6-\tfrac{38}{75},6-\tfrac{38}{75}\big)$. So, the freezing time is two and the cluster sizes are given by $(5,1,1,1,5)$.
\end{example}

The moral of Lemma~\ref{lemma_seg_bound_1} and Lemma~\ref{lemma_seg_bound_2} is that if we have a split with a resulting small connected component, then at both or the unique other side we have quite many agents condensed in a small interval. In principle this can be used to upper bound the number of connected components.
\begin{lemma}
  \label{lemma_width_upper_bound}
  Let $\cP=(X(0),\varepsilon)$ be a 
  BC-process, where $\omega(0),\varepsilon>0$. Then, we have
  \begin{equation}
    S(\cP) < \omega(0)/\varepsilon+1.
  \end{equation}
\end{lemma}

However, in order to exploit the trivial upper bound in Lemma~\ref{lemma_width_upper_bound}, there are a few technical obstacles to overcome. E.g., in Lemma~\ref{lemma_seg_bound_2} feasible parameters are given by $\left(s_1,s_2,a_1,a_2,b_1,b_2\right)=(1,1,3,0,0,3)$, so that there are seven agents within an interval of length $4$. As shown in Example~\ref{ex_3_1_3} we can have connected components of size $3$ next to a connected component of size $1$. However, the three agents on the left and on the right do not need to converge to a single connected component, see Example~\ref{ex_three_ones}. So, if small connected components split, then there are dense regions of agents which can only be split by even denser regions of agents.

\begin{example}
 Consider a BC-process with $\varepsilon=1$, $x_i(0)\in \Z$ for all $i$ and there is a unique agent with starting opinion $0$, unique agents with starting opinion $\pm 1$, two agents with starting opinion $\pm 2$, three agents with starting opinion $\pm 3$, six agents with starting opinion $\pm 4$, ten agents with starting opinion $\pm 5$, and $19$ agents with starting opinion $\pm 6$. After two time steps we reach a stable state where the sizes of the connected components are given by $(29,6,3,2,1,1,1,2,3,6,29)$.
\end{example}

Of course we can extend the previous example with exponentially growing cluster sizes. Indeed we may conjecture that it is some kind of a worst case example.
\begin{conjecture}
  \label{conj_1}
  Let $\left(c_1,\dots,c_l\right)$ be a sequence of cluster sizes of a connected BC-process. Then, there cannot be a contiguous subsequence $(1,1,1,1)$ or $(2,2,1,1,1,2,2)$.
\end{conjecture}

So, while there may be a lot of small cases to check, see Lemma~\ref{lemma_seg_bound_1} and Lemma~\ref{lemma_seg_bound_2} for examples, it might be a finite problem to determine an exact formula for $S(n)$ or to determine its fractional limit $\lim_{n\to\infty} S(n)/n$ at the very least.

\begin{conjecture}
   $$
     \lim\limits_{n\to\infty}\frac{S(n)}{n}=\frac{1}{2}
   $$
\end{conjecture}

We remark that Conjecture~\ref{conj_1} implies $\tfrac{S(n)}{n}\le \tfrac{4}{5}+\tfrac{3}{5n}$. Superadditivity of $S(\cdot)-1$, i.e.\ $S(m+n)\ge S(m) + S(n)-1$, would
imply the existence of the limit of $S(n)/n$. However, even the monotonicity property $S(n+1)\ge S(n)$ is currently open. A proof for $S(n)\le n-2$, where $n\ge 3$, is relatively easy.

\begin{example}
  \label{ex_3_1_3}
  Let $\cP=(X(0),1)$ be a BC-process with
  $$
    X(0)=\big(0,1,2,2.5,3,4,5\big),
  $$
  so that
  $$
    X(1)=(0.5,1,2.125,2.5,2.875,4,4.5)
  $$
  and
  $$
    X(2)=(0.75,0.75,2.5,2.5,2.5,4.25,4.25).
  $$
  Thus, we have $S(7)\ge 3$ and a vector $(2,3,2)$ of cluster sizes. For
  $$
    X(0)=\big(0,0,1,2,3,4,4\big)
  $$
  we compute
  $$
    X(1)=\left(\tfrac{1}{3},\tfrac{1}{3},\tfrac{3}{4},2,\tfrac{13}{4},\tfrac{11}{3},\tfrac{11}{3}\right)
  $$
  and
  $$
    X(2)=\left(\frac{17}{36},\frac{17}{36}, \frac{17}{36}, 2, \frac{127}{36}, \frac{127}{36}, \frac{127}{36}\right)
  $$
  ending up with cluster sizes $(3,1,3)$.
\end{example}
Indeed, we can computationally check $S(7)=S(8)=3$.

\begin{example}
  \label{ex_s_9_ge_4}
  Let $\cP=(X(0),1)$ be a BC-process with
  $$
    X(0)=\big(0,1,2,2.5,3,4,5,5.5,6\big),
  $$
  so that
  $$
    X(1)=(0.5,1,2.125,2.5,2.875,4,5.125,5.5,5.5)
  $$
  and
  $$
    X(2)=(0.75,0.75,2.5,2.5,2.5,4,5.375,5.375,5.375).
  $$
  Thus, we have $S(9)\ge 4$ and a vector $(2,3,1,3)$ of cluster sizes.
\end{example}
Indeed, we can computationally check $S(9)=4$ and cluster size vector $(2,3,1,3)$ as well as the mirrored version $(3,1,3,2)$ are the only possibilities for four clusters.
For $10$ agents and four clusters the following size vectors can occur up to mirror symmetry: $(2,1,5,2)$, $(2,3,1,4)$, $(2,3,2,3)$, $(2,3,3,2)$, $(2,4,1,3)$, $(3,1,1,5)$, and $(3,1,3,3)$. Moreover, we indeed have $S(10)=4$.

\section{Conclusion and outlook}\label{sec_conclusion}

Here we studied different notions of equivalence classes for the BC model. We mainly focused on the combinatorial description of BC-processes as sequences of influence graphs. Proposing linear programming as a computational tool to decide the existence of a suitable starting configuration yields a full machinery to study properties of the basically infinite number of BC-processes by restricting to a finite number of equivalence classes. Here we considered the maximum number of time steps of a dynamic to end up in a stable state as well as the maximum fragmentation. In the literature also the width of the final configuration is a key characteristic. However, this value indeed depends on the precise starting configuration, so that we did not study it in this paper.

For the maximum freezing time $F(n)$ it would be interesting to compute exact values for larger numbers of agents, or bounds at the very least. Here more research on algorithmic details and improvements for our presented general framework are needed. Maybe the dumbbell construction from \cite{wedin2015quadratic} can be refined a bit to also yield relatively large freezing times for medium sized numbers of agents. For the maximum fragmentation $S(n)$ the ultimate goal is to find an exact formula, which does not fully seem out of reach.

So far, we only asked for extreme cases in terms of a given number of agents. A natural next step might be to study different variants of non-monotonicity and counterintuitive behavior described in the literature. It is well known that slightly increasing the confidence level $\varepsilon$ can increase the freezing time and the number of clusters. For large $\varepsilon$ there is only one cluster and the freezing time equals one unless $X(0)$ is constant. But how extreme can these non-monotonicities be? In \cite{hegselmann2023bounded} the author reported several instances where slightly increasing the confidence level $\varepsilon$ turned a consensus into three different opinion clusters.\footnote{Recently, an example where slightly increasing $\varepsilon$ turns a consensus into five clusters was found by R.~Hegselmann (personal communication).}

Of course, similar questions can be also studied for higher dimensional BC-processes, where even the characterization of possible influence graphs might be an interesting problem.


\end{document}